\documentclass[11pt]{article}
\usepackage[utf8]{inputenc}
\usepackage[boxed,linesnumbered,vlined]{algorithm2e}
\usepackage{palatino}
\usepackage{amssymb}
\usepackage{amsfonts}
\usepackage{amsmath}
\usepackage{amsthm}
\usepackage{latexsym}
\usepackage{color}
\usepackage{epsfig}
\usepackage{hyperref}
\usepackage{bm}
\usepackage{comment}
\usepackage{thmtools}
\usepackage{thm-restate}

\newcommand{\ignore}[1]{}
\newenvironment{proofof}[1]{\par{\noindent \bf Proof of #1:}}{\qed\par}

\makeatletter
\renewcommand{\subsubsection}{\@startsection{subsubsection}{3}{0pt}{-12pt}{-5pt}{\normalsize\bf}}
\makeatother

\newcommand{\red}[1]{{{#1}}}
\newtheorem{claim}{Claim}[section]

\newtheorem{remark}[claim]{Remark}
\newtheorem{lemma}[claim]{Lemma}
\newtheorem*{theorem*}{Theorem}
\newtheorem*{lemma*}{Lemma}
\newtheorem{theorem}[claim]{Theorem}
\newtheorem{definition}[claim]{Definition}
\newtheorem{corollary}[claim]{Corollary}

\newtheorem*{definition*}{Definition}

\newcommand{\R}{{\bf R}}

\newcommand{\bX}{{\bf X}}
\newcommand{\bx}{{\bf x}}

\newcommand{\by}{{\bf y}}
\newcommand{\bz}{{\bf z}}

\newcommand{\bZ}{{\bf Z}}

\newcommand{\dist}{\mathrm{dist}}
\newcommand{\Inf}{\mathsf{Inf}}

\newcommand{\E}{\mathop{\mathbf{E}}}

\newcommand{\Var}{\mathrm{Var}}

\newcommand{\eps}{\epsilon}

\newcommand{\sgn}{\mathrm{sgn}}

\newcommand{\poly}{\mathrm{poly}}

\newcommand{\res}{\upharpoonright}

\newcommand{\Has}[1]{\mathsf{Has}_{#1}}

\newcommand{\coords}{\mathcal{S}}
\newcommand{\othercoords}{\mathcal{T}}
\newcommand{\Dict}{\mathsf{Dict}}
\newcommand{\dictset}{\mathcal{D}}

\newcommand{\junta}{\mathcal{J}}
\newcommand{\infsamp}{\mathcal{X}}

\newcommand{\Unif}{\mathsf{Unif}}
\newcommand{\favg}{f_{\mathsf{avg},\coords}}
\newcommand{\gavg}{g_{\mathsf{avg},\coords'}}
\newcommand{\fsmooth}{f_{\mathsf{smooth}}}
\newcommand{\fsmoothavg}{f_{\mathsf{smooth},\mathsf{avg},\coords'}}
\newcommand{\ffinal}{f_{\mathsf{smooth},\mathsf{avg},\tilde \coords'}}
\newcommand{\randeta}{\mathbf{\eta}} 

\begin{document}
 \title{Junta-Correlation is Testable}
\author{Anindya De\thanks{Supported by NSF grant CCF-1814706}\\ 
University of Pennsylvania \\{\tt anindyad@cis.upenn.edu} \and Elchanan Mossel\thanks{Supported by ONR grant N00014-16-1-2227   and 
NSF grant CCF 1320105.} \\ MIT \\ {\tt elmos@mit.edu} \and Joe Neeman\thanks{Supported by the Alfred P.\ Sloan Foundation} \\ UT Austin \\ {\tt jneeman@math.utexas.edu}}
 \maketitle

 \begin{abstract}
The problem of tolerant junta testing is a natural and challenging problem which asks if the property of a function having some specified correlation with a $k$-Junta is testable. 
In this paper we give an affirmative answer to this question: 
We show that given distance parameters $\frac{1}{2} >c_u>c_{\ell} \ge 0$, there is a  tester
which given oracle access to $f:\{-1,1\}^n \rightarrow \{-1,1\}$, with query complexity $ 2^k \cdot \mathsf{poly}(k,1/|c_u-c_{\ell}|)$ and distinguishes between the following cases: 
\begin{enumerate}
\item The distance of $f$ from any $k$-junta is at least $c_u$;
\item There is a $k$-junta $g$ which has distance at most $c_\ell$ from $f$. 
\end{enumerate}
This is the first non-trivial tester (i.e., query complexity is independent of $n$) which works for all $1/2 > c_u > c_\ell \ge 0$. 
The best previously known results by Blais \emph{et~ al.}, required $c_u \ge 16 c_\ell$.   In fact, with the same query complexity, we accomplish the stronger goal of identifying the most correlated $k$-junta, up to permutations of the coordinates. 

We can further improve the query complexity to $\mathsf{poly}(k, 1/|c_u-c_{\ell}|)$ for the (weaker) task of distinguishing between the following cases: 
\begin{enumerate}
\item The distance of $f$ from any $k'$-junta is at least $c_u$. 
\item There is a $k$-junta $g$ which is at a distance at most $c_\ell$ from $f$. 
\end{enumerate}
Here $k'=O(k^2/|c_u-c_\ell|^2)$. 

Our main tools are Fourier analysis based algorithms that simulate oracle access to influential coordinates of functions.

\end{abstract}

\newpage
\section{Introduction}~\label{sec:intro}
 Juntas are a fundamental class of functions in Boolean function analysis. A function $f: \{-1,1\}^n \rightarrow \{-1,1\}$ is said to be a $k$-junta if there are some $k$-coordinates $i_1, \ldots, i_k \in [n]$ such that $f(x)$ only depends on $x_{i_1}, \ldots, x_{i_k}$. In particular, special attention has been devoted to the problem of testing juntas. 
 
 We recall that a property testing algorithm for a class of functions $\mathcal{C}$  is an algorithm which  given oracle access to an $f:\{-1,1\}^n \rightarrow \{-1,1\}$ and a distance parameter $\epsilon>0$, satisfies 
 \begin{enumerate}
 \item If $f \in \mathcal{C}$, then the algorithm accepts with probability at least $2/3$; 
 \item If $\mathsf{dist}(f,g) \ge \epsilon$ for every $g \in \mathcal{C}$, then the algorithm rejects with probability at least $2/3$. Here $\mathsf{dist}(f,g) = \Pr_{x \in \{-1,1\}^n} [f(x) \not =g (x)]$. 
 \end{enumerate}
The principal measure of the efficiency of the algorithm is its \emph{query complexity}. Also, the precise value of the confidence parameter is irrelevant and $2/3$ can be replaced by any constant $1/2 < c<1$. 

Fischer \emph{et~al.}\cite{FKRSS03} were 
the first to study the problem of testing $k$-juntas 
and showed that $k$-juntas can be tested with query complexity $\tilde{O}(k^2/\epsilon)$. The crucial feature of their algorithm is that the query complexity is independent of the ambient dimension $n$. Since then, there has been a long line of work on testing juntas~\cite{blais2009testing,blais2008improved,servedio2015adaptivity,Chen:2017:SQC,CLSSX18} and it continues to be of interest down to the present day. The flagship result here is that $k$-juntas can be tested with $\tilde{O}(k/\epsilon)$ queries and this is tight~\cite{blais2009testing,Chen:2017:SQC}. While the initial motivation to study this problem came from long-code testing~\cite{belgolsud98, PRS02} (related to PCPs and inapproximability), another strong motivation comes from the \emph{feature selection} problem in machine learning (see, e.g.~\cite{Blum:94, BlumLangley:97}). 

\paragraph*{Tolerant testing} The definition of property tester above requires the algorithm to accept if and only if $f \in \mathcal{C}$. However, for many applications, it is important consider a \emph{noise-tolerant} definition of property testing. In particular, Parnas, Ron and Rubinfeld~\cite{parnas2006tolerant} introduced the following definition of noise tolerant testers. 
\begin{definition}
For constants $1/2>c_u> c_{\ell} \geq 0$ and a function class $\mathcal{C}$, a $(c_u,c_{\ell})$-noise tolerant tester  for $\mathcal{C}$ is an algorithm which given oracle access to a function $f: \{-1,1\}^n \rightarrow \{-1,1\}$ 
\begin{enumerate}
\item accepts with probability at least $2/3$ if 
$\min_{g \in \mathcal{C}} \mathsf{dist}(f,g) \le c_\ell$. 
\item  rejects with probability at least $2/3$ if 
$\min_{g \in \mathcal{C}} \mathsf{dist}(f,g) \ge c_u$. 
\end{enumerate}
\end{definition}
We observe that we restrict $c_u, c_\ell <1/2$. This is because most natural classes $\mathcal{C}$ are closed under complementation -- i.e., if $g \in \mathcal{C}$, then $-g \in \mathcal{C}$. For such a class $\mathcal{C}$ and for any $f$,  $\min_{g \in \mathcal{C}} \mathsf{dist}(f,g) \le 1/2$. 
Further, note that the standard notion of property testing corresponds to a $(\epsilon,0)$-noise tolerant tester.


 The problem of testing juntas becomes quite challenging in the presence of noise. Parnas \emph{et~al.}~\cite{parnas2006tolerant} observed that 
 any tester whose (individual) queries are uniformly distributed are inherently noise tolerant in a very weak sense. In particular, \cite{DLM+:07} used this observation to show that the junta tester of \cite{FKRSS03} is in fact a $(\epsilon, \mathsf{poly}(\epsilon/k))$-noise tolerant tester for $k$-juntas -- note that $c_\ell$ is quite small, namely $\mathsf{poly}(\epsilon/k)$. Later, Chakraborty~\emph{et al.}~\cite{chakraborty2012junto} showed that the tester of Blais~\cite{blais2009testing} yields a $(C\epsilon,\epsilon)$ tester (for some large but fixed $C>1$) with query complexity  
  $\mathsf{exp}(k/\epsilon)$. Recently, there has been a surge of interest in tolerant junta testing. On one hand, Levi and Waingarten showed that there are constants $1/2>\epsilon_1>\epsilon_2>0$ such that any non-adaptive $(\epsilon_1, \epsilon_2)$ tester requires $\tilde{\Omega}(k^2)$ non-adaptive queries. Contrast this with the result of Blais~\cite{blais2009testing} who showed that there is  non-adaptive tester for $k$-juntas with $O(k^{3/2})$ queries when there is no noise. In particular, this shows a gap between testing in the noisy and noiseless case. 

In the opposite (i.e., algorithmic) direction, Blais~\emph{et~al.}~\cite{blais2018tolerant}  proved the following theorem. 
\begin{theorem}~\label{thm:blais2018}
There is an algorithm which for any $\rho \in (0,1)$, $\epsilon \in (0,1)$ and parameter $k \in \mathbb{N}$, is a $(\epsilon, \frac{\rho \epsilon}{16})$-noise tolerant tester for $k$-juntas. The query complexity of the tester is $O\big( \frac{k \log k}{\epsilon \rho (1-\rho)^k}\big)$.
\end{theorem}
Note that for any $C>16$, this yields an $(\epsilon, \epsilon/C)$-tolerant tester for $k$-juntas with $\exp(k)$ query complexity. Thus, it improves on the result of \cite{chakraborty2012junto} who showed the same result for an unspecified large constant $C$. 

To understand the main shortcoming of  \cite{blais2018tolerant}, note that this algorithm does not yield a $(c_u, c_\ell)$ noise tolerant tester once $c_\ell > \frac{1}{32}$ -- e.g, no setting of parameters in the tester of \cite{blais2018tolerant} can yield (say) a $(0.1, 0.05)$ noise-tolerant tester for $k$-juntas. Naturally, one would like to obtain $(c_u, c_\ell)$ testers for any $1/2> c_u > c_\ell$. The main result of this paper  accomplishes  this goal. Below we formalize and state the main results of the paper. 

We will use $\mathcal{J}_{n,k}$ to denote the class of $k$-juntas on $n$ variables.  Also, for a subset $S\subseteq [n]$, we let $\mathcal{J}_{S,k}$ denote the class of $k$-juntas on the variables in $S$. Further, unless indicated otherwise, all expectations are taken over uniformly random elements of $\{-1,1\}^n$ where the ambient dimension $n$ will be clear from the context. Our first result constructs $(c_u, c_\ell)$ testers for any $1/2> c_u > c_\ell$. 
\begin{theorem}~\label{thm:main-tester}
There is an algorithm \textsf{Maximum-correlation-junta} which takes as input parameters $k \in \mathbb{N}$, distance parameter $\epsilon>0$, oracle access to function $f: \{-1,1\}^n \rightarrow \{-1,1\}$  and has the   following guarantee: With probability $2/3$, it outputs a number $\widehat{\mathsf{Corr}}_{f,k}$ such that 
\[
\big| \widehat{\mathsf{Corr}}_{f,k} - \max_{\ell \in \mathcal{J}_{n,k}} \E_{\bx}[\ell(\bx) \cdot f(\bx)] \big| \le \epsilon. 
\]
It also outputs a function $h: \{-1,1\}^k \rightarrow \{-1,1\}$ such that there is a set of coordinates $\mathcal{T} = \{i_1, \ldots, i_k\} \subseteq [n]$ and 
\[
\big|\max_{\ell \in \mathcal{J}_{n,k}} \E_{\bx}[\ell(\bx) \cdot f(\bx)] -   \E_{\bx}[h(\bx_{i_1}, \ldots, \bx_{i_k}) \cdot f(\bx)]\big|  \leq \epsilon. 
\]
The query complexity of the algorithm is $ 2^k \cdot \mathsf{poly}(k,1/\epsilon)$.
\end{theorem}
Note that the above algorithm is doing something stronger than ``merely"
computing correlation of $f$ with $k$-juntas -- in fact, the algorithm also
outputs the a  $k$-junta that is most correlated up to $\eps$. Note that the
algorithm cannot identify the actual subset of the coordinates of $f$ that maps
to those of the junta, as an standard information theory argument shows that this
requires the number of queries to depend on $n$, even without noise. Further,
for the task of approximately finding the most correlated $k$-junta, our query complexity is essentially optimal,
since even giving the description of the most correlated $k$-junta
takes $2^k$ bits.  An immediate
corollary of Theorem~\ref{thm:main-tester} is the existence of a noise tolerant
tester for $k$-juntas.

\begin{corollary}~\label{corr:testing-main}
For any constant $\frac{1}{2} >c_u>c_\ell\ge 0$ and $k \in \mathbb{N}$, there is $(c_u, c_\ell)$-noise tolerant tester for $k$-juntas with query complexity 
$2^k \cdot \mathsf{poly}(k, 1/|c_u-c_\ell|)$. 
\end{corollary}
\begin{proof}
Let $\epsilon = \frac{c_u-c_\ell}{2}$. Run the algorithm \textsf{Maximum-correlation-junta} with distance parameter $\epsilon$. Let the output be $\widehat{\mathsf{Corr}}_{f,k}$. Set $\mathsf{Thr} = 1-2c_\ell - 2 \epsilon = 1 - 2c_u + 2\epsilon$.  
The rest of the algorithm is 
\begin{enumerate}
\item If $\widehat{\mathsf{Corr}}_{f,k} \ge \mathsf{Thr}$, then the algorithm accepts.
\item The algorithm rejects otherwise. 
\end{enumerate}
Note that if there is a $k$-junta $g$ such that $\mathsf{dist}(f,g) \le c_\ell$, then $\max_{g \in \mathcal{J}_{n,k}} \E_{\bx}[g(\bx) \cdot f(\bx)] \ge 1-2c_\ell$. Thus, $\widehat{\mathsf{Corr}}_{f,k} \ge 1-2c_\ell-\epsilon$ (w.p. $2/3$),
and so the algorithm will accept.

On the other hand, if $\dist(f, g) \ge c_u$ for every $g \in \junta_{n,k}$ then
$\max_{g \in \junta_{n,k}} \E_{\bx}[g(\bx)f(\bx)] \le 1 - 2 c_u = \mathsf{Thr} - 2\epsilon$, meaning
that the algorithm will reject with probability at least $2/3$.
\end{proof}
We also remark here that the algorithm \textsf{Maximum-correlation-junta} can be modified in a straightforward manner to yield a noise tolerant tester against any \emph{subclass of juntas}, including any \emph{specific junta} -- e.g., for any $1/2>c_u>c_\ell \ge 0$, we can obtain a $(c_u, c_\ell)$-tester for $k$-linear functions~\cite{blais2012tight, Saglam2018} with query complexity $\poly(k,1/|c_u-c_\ell|)$. We leave the proof to the interested reader.

Finally, we can also improve the query complexity to have a polynomial dependence on $k$ at the cost of achieving a weaker guarantee. 
\begin{theorem}~\label{thm:main-tester-gap}
There is an algorithm \textsf{Maximum-correlation-gap-junta} which takes as
input parameters $k \in \mathbb{N}$, distance parameter $\epsilon>0$, oracle
access to a function $f: \{-1,1\}^n \rightarrow \{-1,1\}$ and has the following
guarantee:  With probability $2/3$, it outputs a number
$\widehat{\mathsf{Corr}}_{f,\mathsf{gap},k}$ satisfying
\[
    \max_{g \in \junta_{n,k}} \E_{\bx}[g(\bx) f(\bx)] - \epsilon \le \widehat{\mathsf{Corr}}_{f,\mathsf{gap},k}
    \le \max_{g \in \junta_{n,k^2/\epsilon^2}} \E[g(\bx)f(\bx)] + \epsilon.
\]
The query complexity of the algorithm is $\poly(k, 1/\epsilon)$.
\end{theorem}
Analogous to  Corollary~\ref{corr:testing-main}, we get the following corollary by applying Theorem~\ref{thm:main-tester-gap}. 
\begin{corollary}~\label{corr:tester-gap}
For any constant $\frac{1}{2} >c_u>c_\ell\ge 0$ and $k \in \mathbb{N}$, there is an algorithm with query complexity 
$\mathsf{poly}(k, 1/|c_u-c_\ell|)$ which with probability $2/3$ can distinguish between the following 
cases: 
\begin{enumerate}
\item $\min_{g \in \mathcal{J}_{n,k}} \mathsf{dist}(g,f) \le c_\ell$. 
\item $\min_{g \in \mathcal{J}_{n,k'}}\mathsf{dist}(g,f) \ge c_u$, where $k' = k^2/(c_u - c_\ell)^2$.
\end{enumerate}
\end{corollary}

We remark that~\cite{blais2018tolerant} contains a result along the same lines
as Corollary~\ref{corr:tester-gap}, but with $k' = 4k$ and $c_u = 16 c_\ell$. That is,
compared to the result of~\cite{blais2018tolerant}, Corollary~\ref{corr:tester-gap}
has a worse $k'$, but allows for arbitrarily good noise tolerance.
 
\subsection{Overview of techniques}~\label{sec:overview}
One of our main contributions is a {\em Fourier based algorithm}
to simulate {\em oracles} to interesting coordinates of $f$. 
In particular, the first step in both \textsf{Maximum-correlation-junta}
and \textsf{Maximum-correlation-gap-junta} is to obtain oracle access
to the functions $x \mapsto x_\ell$ for all $\ell$ with large
low-degree influence in $f$. This idea previously appeared
(but in a real-valued setting) in~\cite{DMN:18}, and it may
be of independent interest. In particular, it is a substantial
departure from previous approaches to tolerant junta testing.

\subsubsection{Sketch of algorithm  \textsf{Maximum-correlation-junta}.}
We begin by giving the high level overview of the algorithm \textsf{Maximum-correlation-junta} (from Theorem~\ref{thm:main-tester}). 
Let $f: \{-1,1\}^n \rightarrow \{-1,1\}$ and let $g:\{-1,1\}^n \rightarrow \{-1,1\}$ be a maximally correlated $k$-junta -- for this description, assume that  $\mathbf{E}_{\bx} [f(\bx) \cdot g(\bx)] =\Omega(1)$. The main steps in our algorithm is as follows. 

\begin{enumerate}
\item First, we show that up to a small loss in correlation, we may assume that
    every variable in $g$ has at least $k^{-\Theta(1)}$ influence in $f$ -- see
    Claim~\ref{clm:correlation1} for the precise quantitative parameters.
    We call these the ``interesting variables,'' and our first main goal
    is to obtain oracle access to them.

\item Suppose that $x_\ell$ is an interesting variable. We show that
    by randomly (in a carefully chosen sense) restricting certain variables
    of $f$, with probability $k^{-O(1)}$ we obtain a function (call it $f_{\res \rho}$)
    such that $|\widehat{f_{\res \rho}}(\ell)| \ge k^{-O(1)}$.
    (See Claim~\ref{clm:restrict} for the precise details.)
    In other words, influential coordinates of $f$ end up with large Fourier coefficients
    under random restrictions.
    \label{step:restrict}

\item Assuming that $|\widehat{f_{\res \rho}}(\ell)| \ge k^{-O(1)}$, we construct a (randomized)
    operator on the function $f_{\res \rho}$ which, with probabililty $k^{-O(1)}$, gives us
    an oracle to the variable $x_\ell$.
    This operator is a variant of the operator used by
    H{\aa}stad~\cite{Hastad}) in the context of dictatorship testing and in
    turn uses a modified version of the standard Bonami-Beckner noise. The
    details of this are in Section~\ref{sec:single-oracle}. 
    \label{step:Hastad}

\item Having obtained an oracle for one particular variable $x_\ell$, we can
    just repeat steps~\ref{step:restrict} and~\ref{step:Hastad} $k^{\Theta(1)}$ times to obtain a set
    $\coords$ of (oracles to) variables that contains all of the interesting variables.
    This reduces the original problem (estimating $\max_{g \in \junta_{n,k}} \E_{\bx}[g(\bx)f(\bx)]$)
    to the problem of estimating $\max_{g \in \junta_{\coords,k}} \E_{\bx}[g(\bx)f(\bx)]$.
    We do this via a simple sampling based algorithm in
    \textsf{Find-best-fit}. The query complexity of this routine is  $2^k \cdot
    \mathsf{poly}(k)$ and is the bottleneck for
    \textsf{Maximum-correlation-junta}.  
\end{enumerate}  

\subsubsection{Sketch of Algorithm \textsf{Maximum-correlation-gap-junta}}
The difference in the proof of Theorem~\ref{thm:main-tester-gap} vis-a-vis Theorem~\ref{thm:main-tester} lies in Step~4 of the above overview. Namely, having obtained the set $\mathcal{S}$, our goal is find smaller subset $\mathcal{S'} \subset \mathcal{S}$ of size $O(k^2/\eps^2)$ of the variables that achieves the same correlation and moreover find this correlation. 
\begin{enumerate}
\item The novel idea of the proof is to find a polynomial algorithm that is able to compute the function 
\[
f_{\mathsf{avg}, \mathcal{S}}(x) :=
\mathbf{E}_{\by \in \{-1,1\}^{[n] \setminus \mathcal{S}}} [f(x,\by)]
\]
{\red{while only having oracle access to the variables in $\mathcal{S}$}}.
The details of this algorithm are in Section~\ref{sec:averaging}, and we will
give an outline shortly.
    Note that $f_{\mathsf{avg}, \mathcal{S}}$ depends only on the variables in $\coords$ (of which there are $\poly(k)$), and among all such
    functions it has the highest correlation with $f$. Further,
    $\max_{g \in \junta_{\coords,k}} \E_{\bx}[g(\bx)f(\bx)] = \max_{g \in \junta_{\coords,k}} \E_{\bx}[g(\bx)f_{\mathsf{avg}, \mathcal{S}}(\bx)]$.

\item
    We replace $f_{\mathsf{avg}, \mathcal{S}}$ by $T_{\rho} f_{\mathsf{avg}, \mathcal{S}}$ (for $\rho = 1 - O(\eps/k)$), incurring an $O(\epsilon)$ error in the correlation.
    The advantage of $T_{\rho} f_{\mathsf{avg}, \mathcal{S}}$ over $f_{\mathsf{avg}, \mathcal{S}}$ is that it has at most $O(k^2/\eps^2)$
    high-influence (meaning, influence $\Omega(\eps/k)$) variables,
    and that restricting our attention to juntas on these variables
    only incurs another $O(\epsilon)$ error in the correlation.
    It is also easy to produce an oracle to $T_{\rho} f_{\mathsf{avg}, \mathcal{S}}$ from $f_{\mathsf{avg}, \mathcal{S}}$ with polynomially many samples. 

\item
    Our next step is to estimate the influence (in the function $T_\rho f_{\mathsf{avg}, \mathcal{S}}$) of all the variables in $\mathcal{S}$.
    We can do so by sampling correlated pairs $(x,y)$ repeatedly until we obtain
    pairs that differ in one coordinate and then checking the effect on $T_{\rho} f_{\mathsf{avg}, \mathcal{S}}$. 
    Having estimated the influences, we let $\coords' \subseteq \coords$ be
    the set of high-influence variables.
    The problem is reduced to that of estimating $\max_{g \in \junta_{\coords',k}} \E_{\bx}[g(\bx)f(\bx)]$.

\item 
    The final output of the algorithm is an estimate for the correlation of $f$ with the best function
    depending only on variables of $\coords'$. This is just
    \[
        \E_\bx[|f_{\mathsf{avg},\coords'}(\bx)|],
    \]
    which can be estimated using the same averaging procedure that we mentioned in the first step.
\end{enumerate}

It remains to explain how to carry out the averaging procedure needed for the first and last steps
of the outline above: how do we estimate $\E_{\by \in \{\pm 1\}^{[n] \setminus \coords}} [f(x, \by)]$,
given only oracle access to the variables in $\coords$?
The basic idea is to perform a random walk on the the subset of $y$'s that agree with $x$ 
on all the elements of $\mathcal{S}$. We let $y^{(0)} = x$. Given $y^{(i)}$ we sample $y^{(i+1)}$ to be a noisy version of $y^{(i)}$, where each coordinate is flipping with probability about $1/|\mathcal{S}|$. 
We {\em accept} $y^{(i+1)}$ if the value of all oracle functions in $\mathcal{S}$ is identical for $y^{(i+1)}$ and 
$y^{(i)}$. If we reject the current proposal of $y^{(i+1)}$ we try again an independent noisy $y^{(i+1)}$.

Thus we effectively perform a random walk on the noisy hyper-cube on the
coordinates in $[n] \setminus \mathcal{S}$. The spectral gap of the random walk
is inverse polynomial in $k$, and hence by taking $\poly(k)$ steps of this random
walk, we can essentially independently resample those coordinates of $x$
that do not belong to $\coords$. By repeating this, we can estimate $f_{\mathsf{avg},\coords}(x)$.

\section{Preliminaries}~\label{sec:preliminaries}
We begin with the basics of Fourier analysis, in particular the notion of Fourier expansion of functions. 
\begin{definition}~\label{def:Fourier-expansion}
For any subset $S \subseteq [n]$, we define $\chi_S: \{-1,1\}^n \rightarrow \{-1,1\}$ as $\chi_S(x) = \prod_{i \in S} x_i$. Any function $f: \{-1,1\}^n \rightarrow \mathbb{R}$ can be expressed as a linear combination of $\{\chi_S(x)\}_{S \subseteq [n]}$ (as follows): 
\[
f(x) = \sum_{S \subseteq [n]} \widehat{f}(S) \chi_S(x). 
\]
This is referred to as the Fourier expansion of $f$ and the coefficients $\{\widehat{f}(S)\}_{S \subseteq [n]}$ 
are referred to as the Fourier coefficients of $f$. 
\end{definition}
We next define the concept of influence of variables in $f: \{-1,1\}^n \rightarrow \mathbb{R}$ 
\begin{definition}~\label{def:influence}
For any function $f: \{-1,1\}^n \rightarrow \{-1,1\}$ and any $i \in [n]$, $\mathsf{Inf}_i(f) = \Pr_{x \in \{-1,1\}^n} [f(x) \not = f(x^{\oplus i})]$ (where $x^{\oplus i}$ differs from $x$ exactly in the $i^{th}$ position). 
In terms of Fourier coefficients, $\Inf_i(f) = \sum_{S \ni i} \widehat{f}^2(S)$; in the case of
a real-valued function $f: \{-1,1\}^n \rightarrow \mathbb{R}$, we take this latter formula as the definition of $\Inf_i(f)$.

For a number $k \le n$, we let $\mathsf{Inf}_i^{\le k}(f)$ denote the quantity $\mathsf{Inf}_i^{\le k}(f) = \sum_{S \ni i: |S| \le k} \widehat{f}^2(S)$.  We also define the total influence of $f$, denoted by $\mathsf{Inf}(f)$, as $\sum_{S} |S| \widehat{f}^2(S)$. 
\end{definition}

We now define the Bonami-Beckner noise operator on the space of functions on $\{-1,1\}^n$. To do this, we first define a general notion of noise distribution on $\{-1,1\}^n$. 
For $\eta \in [-1,1]^n$, we let $\mathbf{Z}_\eta$ denote the product distribution on $\{-1,1\}^n$ where the expectation of the $i^{th}$ bit is $\eta_i$. 
\begin{definition}
For any $\rho \in [-1,1]$, let $\overline{\rho} \in [-1,1]^n$ denote the vector all of whose coordinates are $\rho$. 
The Bonami-Beckner noise operator (denoted by $T_{\rho}$) operates on $f: \{-1,1\}^n \rightarrow \mathbb{R}$ as 
\[
T_{\rho} f(x) = \mathbf{E}_{\by \sim \mathbf{Z}_{\overline{\rho}}} [f(x \cdot y)]. 
\]
We let $x \cdot y$ denote the coordinate wise product of $x$ and $y$.  
\end{definition}
A standard fact about the operator $T_{\rho}$ is its action on the Fourier expansion of $f$ (see~\cite{ODonnell:book} for details). 
\[
T_{\rho} f(x) = \sum_{S \subseteq [n]} \rho^{|S|} \widehat{f}(S) \chi_S(x). 
\]
\ignore{
\begin{claim}~\label{clm:total-inf}
Let $f: \{-1,1\}^n \rightarrow [-1,1]$. Then, $\mathsf{Inf}(T_{1-\delta} f) \le 1/\delta$. 
\end{claim}
\begin{proof}
It is easy to see that  
$\max_{t \in \mathbb{N}} t \cdot (1-\delta)^t  \le 1/\delta$. 
\[
\mathsf{Inf}(T_{\delta} f)  = \sum_{S} |S| \cdot (1-\delta)^{2|S|} \cdot \widehat{f}^2(S) \le 
 \max_{S} |S| \cdot (1-\delta)^{2|S|} \le 1/\delta. 
\]
\end{proof}}

\subsubsection*{Bonami-Beckner noise operator as a Markov chain}
It will be useful for us to view the Bonami-Beckner noise operator as a Markov chain. We recall the definition of a Markov chain (on a finite set). 
\begin{definition}~\label{def:mchain}
Let $G$ be a finite set and let $P \in \mathbb{R}^{G \times G}$ be a stochastic matrix. 
The random variables (taking values in $G$) $(\bx_i)_{i=1}^T$ are said to follow the Markov chain $\mathcal{M}_P$ (with transition matrix given by $P$) 
 if for any $T \ge j >1$ and any $g_1, \ldots, g_j \in G$, 
\[
\Pr[\bx_j=g_j | \bx_{j-1} = g_{j-1}, \ldots,  \bx_{1} = g_{1}] =\Pr[\bx_j=g_j | \bx_{j-1} = g_{j-1}] = P(g_{j-1}, g_j).  
\]
We refer the reader to the book by Levin and Peres~\cite{levin2017markov} for definitions of standard notions such as ergodicity, aperiodicity and stationary distributions.
\end{definition}
Now, 
consider any $\rho \in [-1,1]$ and define the stochastic matrix $P_{\rho}$ (whose rows and columns are indexed by $\{-1,1\}^n$) such that 
\[
P_{\rho}(x,y) = 
\Pr_{\bz \sim \bZ_{\overline{\rho}}} 
[x \cdot \bz=y]. 
\]
It is easy to see that the matrix $P_{\rho}$ is a symmetric matrix and further, the second largest eigenvalue of the matrix $P_{\rho}$ is at most $\rho$. The matrix $P_{\rho}$ also defines a corresponding Markov chain $\mathcal{M}_{\rho}$ (i.e., the transition matrix of $\mathcal{M}_{\rho}$ is $P_{\rho}$). Markov chains have a certain ``averaging property" which is particularly useful for us and is stated below. We will instantiate it to the Markov chain $\mathcal{M}_{\rho}$ in {Section~\ref{sec:poly-tester} later.} We now state the following result due to Lezaud~\cite{lezaud1998chernoff} (Theorem~1.1 in the paper) which applies to ergodic and reversible Markov chains. Similar results which apply to the special case of random walks on undirected graphs have found many applications in computer science~\cite{gillman1998chernoff, healy2008randomness, vadhan2012pseudorandomness, GLSS18}. 

\begin{lemma}~\label{lem:bias}
    For $\rho \in (-1, 1)$, let $\bx^{(1)}, \bx^{(2)}, \dots$ follow the Bonami-Beckner Markov chain $\mathcal{M}_\rho$
    with an arbitrary initial value $\bx^{(1)} = x \in \{\pm 1\}^n$.
    There is a constant $C$ such that for any $f: \{-1, 1\}^n \to [-1, 1]$ and any $\gamma, \delta > 0$, if
    $T \ge \frac{C \log(1/\delta)}{\gamma^2 \cdot (1-\rho)}$ then
\[
\Pr\bigg[ \bigg| \frac{f(\mathbf{x}_1) + \ldots + f (\mathbf{x}_T)}{T}  - \mathbf{E}_{\bx \sim \{-1,1\}^n} [f(\bx)]\bigg| > \gamma \bigg] \leq \delta. 
\] 
\end{lemma}

\subsection{Random restrictions}
A crucial role in our algorithm will be played by the notion of random restrictions from circuit complexity~\cite{FSS:84, Hastad:86}. 
\begin{definition}~\label{def:restriction}
For any $\mu \in [0,1]$, we let $\mathcal{R}_{\mu} \in \{-1,1,\ast\}^n$ denote
the product distribution where each coordinate is $\ast$ with probability
$\mu$, $\pm 1$ with probability $(1-\mu)/2$ each. Further, for $f: \{\pm 1\}^n
\rightarrow \{\pm 1\}$ and $\xi \in \{-1,0,1\}^n$, we let $f_{\res \xi}:
\{-1,1\}^S \rightarrow \{-1,1\}$ 
where 
\begin{enumerate}
    \item $S = \{i \in [n]: \xi_i = 0\}$. The set of variables in $S$ are said to \emph{survive} in $f_{\res \xi}$. 
\item For $x \in \{\pm 1\}^S$, we $f_{\res \xi}(x)= f(z)$ where $z_S = x$ and $z_j= \xi_j$ for $j \not \in S$. 
\end{enumerate}
\end{definition}
 
\section{Construction of coordinate oracles}~\label{sec:oracle}

The main result of this section is an algorithm for constructing a set of ``oracles'' to
all of the interesting coordinates of a function, assuming that there are not too many of them.
The basic definition is the following: {\red{For $1 \le i \le  n$, let $\Dict_i:\{-1,1\}^n \rightarrow \{-1,1\}$ be defined as $\Dict_i: x \mapsto x_i$. }}
\begin{definition}\label{def:oracle}
    Let $\dictset$ be a set of functions from $\{-1, 1\}^n$ to $\{-1, 1\}$.
    We say that $\dictset$ is an oracle for the coordinates $\coords \subset [n]$
    if
    \begin{itemize}
        \item for every $g \in \dictset$, there is some $i \in \coords$ such that
            $g = \Dict_i$ or $g = -\Dict_i$; and
        \item for every $i \in \coords$, there is some $g \in \dictset$ such that
            $g = \Dict_i$ or $g = -\Dict_i$.
    \end{itemize}
    In other words, $\dictset$ is an oracle for $\coords$ if
    $\dictset = \{\Dict_i: i \in \coords\}$ ``up to sign''.
\end{definition}

Due to our constraints, we will not be able to produce coordinate oracles exactly according
to the definition above, so we will relax it slightly.
Recall that we have fixed an underlying function $f: \{\pm 1\}^n \to \{-1, 1\}$ and
a parameter $k \in \mathbb{N}$.

\begin{definition}\label{def:approximate-oracle}
    Let $\dictset$ be a set of functions from $\{-1, 1\}^n$ to $\{-1, 1\}$.
    For $\epsilon \le \frac 18$, we say that $\dictset$ is an $\nu$-oracle for $\coords \subset [n]$ if
    \begin{itemize}
        \item for every $g \in \dictset$, there is some $i \in \coords$ such that $g$ is $\nu$-close to
            $\Dict_i$ or $-\Dict_i$ (necessarily only one,
          since $\nu \le \frac 18$);
        \item for every $i \in \coords$, there is exactly one $g \in \dictset$ that is $\nu$-close to
            $\Dict_i$ or $-\Dict_i$; and
        \item for every $x \in \{\pm 1\}^n$, every $g \in \dictset$, and every $\delta > 0$, there is a randomized algorithm
            to compute $g(x)$ correctly with probability at least $1-\delta$, using $\poly(k, \log \frac 1\delta)$
            queries to $f$.
    \end{itemize}
\end{definition}
{\red{While the definition of $\mathcal{D}$ involves both $\nu$ and $k$, since the latter will remain fixed throughout, the above definition is only quantified in terms of $\nu$.}}
The parameters $\nu$ and $\delta$ that we choose will essentially allow us to pretend that
an $\nu$-oracle is an oracle. In particular, we will fix $\delta = 2^{-\omega(k)}$ when evaluating coordinate
oracles at a point. This will preserve the $\poly(k)$ query complexity of each oracle query, while ensuring
that (with high probability) every query that we make to a coordinate oracle will be computed correctly
(since in all of our algorithms, we will make no more than $2^{k} \cdot \poly(k)$ queries).
Our choice of $\nu$ will depend on the setting: in the setting of Theorem~\ref{thm:main-tester},
we will make at most $2^{k} \cdot \poly(k)$ queries to each coordinate oracle, so
we will take $\nu = (2^k \cdot \poly(k))^{-1}$. This means that each oracle query requires at most
$\poly(k)$ queries to $f$, while ensuring that each coordinate oracle is so close to a dictator (or anti-dictator)
that we will (with high probability) not observe the difference.
In the setting of Theorem~\ref{thm:main-tester-gap}, we will set take $\nu = \poly(1/k)$ and make
at most $\poly(k)$ queries to each coordinate oracle; this requires $\poly(k)$ queries to $f$ and ensures
that (with high probability) we will not observe the difference between any coordinate oracle
and its corresponding dicator.
With this in mind, and to prevent a proliferation of parameters, we will often pretend that we
have access to an oracle in the sense of Definition~\ref{def:oracle} when we really
have access to an $\nu$-oracle in the sense of Definition~\ref{def:approximate-oracle}.

\subsection{A single oracle}~\label{sec:single-oracle}

We begin by describing how to construct an oracle to a single coordinate.
The basic step notion is the following operator, which
is related to one used by H{\aa}stad in the context of dictatorship testing~\cite{Hastad}.

\begin{definition}~\label{def:Hastad-op}
Let $\eta \in [-1,1]^n$ and $\bZ_{\eta}$ denote the product distribution on $\{-1,1\}^n$ where the expectation of the $i^{th}$ bit is $\eta_i$. For any
$f: \{-1,1\}^n \rightarrow \mathbb{R}$, we define the operator 
\[
\Has{\eta} f(x) = \mathbf{E}_{\by_1, \by_2\in \{-1,1\}^n, \by_3 \in \bZ_{\eta}} [f(\by_1) f(\by_2) f(x \oplus\by_1 \oplus \by_2 \oplus \by_3)]. 
\] 
\end{definition}

In terms of the Fourier expansion, it is easy to check that
\[
\Has{\eta} f(x) = \sum_S \widehat{f}^3(S) \chi_S(x) \eta^S,
\]
where $\eta^S = \prod_{j \in S} \eta_j$. A consequence of this expansion is that for the right choice of $\eta$,
$\Has{\eta} f$ is a good approximation to a certain dictator function (depending on $\eta$).

\begin{lemma}\label{lem:isolate-variable}
    Suppose that $|\hat f(1)| \ge \kappa$, where $\kappa \in (0, 1)$, and let $\alpha = \frac{\kappa^3}{16}$.
    Choose $\randeta \in \{0, \alpha\}^n$ randomly so that $\Pr[\randeta_i = \alpha] = \kappa^{6}/16$,
    independently for every $i$. Then with probability at least $\Omega(\kappa^6)$,
    for every $x \in \{\pm 1\}^n$,
    \[
        \big|\Has{\randeta}f(x) - \widehat{f}^3(0) -\alpha \widehat{f}^3(1) x_1\big| \le \frac{\alpha}{4} |\hat f(1)|^3.
    \]
\end{lemma}

\begin{proof}
    Let $p = \kappa^{6}/16 = \Pr[\randeta_i = \alpha]$.
    Let $\Gamma = \{i: |\hat f(i)| \ge \kappa^3/8\}$; since $\sum_S \hat f(S)^2 \le 1$, we have
    $|\Gamma| \le 64 \kappa^{-6}$. Let $E$ be the event that $\eta_1 = \alpha$
    and $\eta_j = 0$ for all $j \in \Gamma \setminus \{1\}$. Then
    $\Pr[E] = p(1-p)^{|\Gamma| - 1} \ge p(1-p)^{64 \kappa^{-6}} = \Omega(\kappa^6)$, and we will show that the claimed
    inequality happens on $E$. Indeed, the Fourier expansion above implies that
    \begin{align*}
        \Has{\randeta}f(x) - \widehat{f}^3(0) -\alpha \widehat{f}^3(1) x_1
        &= \sum_{1 < j \le n} \eta_j \hat f^3(j) x_j + \sum_{|S| > 1} \eta^S \hat f^3(S) \chi_S(x) \\ \\
        &= \sum_{j \not \in \Gamma} \eta_j \hat f^3(j) x_j + \sum_{|S| > 1} \eta^S \hat f^3(S) \chi_S(x),
    \end{align*}
    where the second equality holds on the event $E$.
    Now,
    \[
        \sum_{j \not \in \Gamma} |\eta_j \hat f^3(j)|
        \le \alpha \frac{\kappa^3}{8} \sum_{j \not \in \Gamma} \hat f^2(j)
        \le \frac{\alpha \kappa^3}{8}
    \]
    and
    \[
        \sum_{|S| > 1} |\eta^S \hat f^3(S)| \le \sum_{|S| > 1} \alpha^2 \hat f^2(S) \le \alpha^2,
    \]
    and so the triangle inequality gives
    \[
        |\Has{\randeta}f(x) - \widehat{f}^3(0) -\alpha \widehat{f}^3(1) x_1|
        \le \frac{\alpha \kappa^3}{8} + \alpha^2 \le \frac{\alpha \kappa^3}{4} \le \frac{\alpha}{4} |\hat f(1)|^3.
        \qedhere
    \]
\end{proof}

Lemma~\ref{lem:isolate-variable} gives us a natural algorithm for computing
something that might be a coordinate oracle: the basic idea is to sample $\randeta$
and then to define
\[
    g_\eta(x) = \sgn(\Has{\randeta} f(x) - \hat f^3(0)) = \sgn(\Has{\randeta} f(x) - \E[f]^3).
\]
Note that computing the function $g_\eta$ requires randomness (to estimate $\Has{\randeta} f(x)$
and $\E[f]$). However, a straightforward Chernoff bound implies that with $O(\kappa^{-12} \log \frac 1\delta)$ queries
to $f$, we can estimate both $\E[f]$ and $\Has{\randeta} f(x)$ to additive accuracy $O(\kappa^6)$, with probability $1-\delta$;
according to Lemma~\ref{lem:isolate-variable}, this is sufficient to correctly evaluate $g_\eta(x)$ with probability
$1 - \delta$. This is the same sort of guarantee required in Definition~\ref{def:approximate-oracle};
as discussed there, we can choose $\delta = 2^{-k^2}$ so that with high probability, every evaluation
of $g_\eta$ will be correct.

Now, Lemma~\ref{lem:isolate-variable} only provides us with a small probability of finding a good $g_\eta$.
To filter out the bad ones, we add in a dictatorship test: for $i \le n$, let $\mathsf{Dict}_i:
\{-1,1\}^n \rightarrow \{-1,1\}$ denote the function $\mathsf{Dict}_i(x) =
x_i$.  We call $\Dict_i$ a dictator function, and $-\Dict_i$ an anti-dictator
function. There is an algorithm for testing whether a function is
a dictator function (see Chapter~7 of~\cite{ODonnell:book}):

\begin{theorem}~\label{thm:dictatorship-test}
 There is an algorithm \textsf{Dictator-test} which given an error parameter {\red{$\nu>0$}} and confidence parameter {\red{$\tilde{\delta}>0$}}, makes $O(\nu^{-1} \log \tilde{\delta}^{-1})$ queries to $f: \{-1,1\}^n \rightarrow \{-1,1\}$ and has the following properties: 
 \begin{enumerate}
 \item If $f: \{-1,1\}^n \rightarrow \{-1,1\}$ is a dictator or an anti-dictator, then it accepts with probability $1$. 
 \item Any $f: \{-1,1\}^n \rightarrow \{-1,1\}$ which is $\nu$-far from every dictator and anti-dictator is accepted with probability at most $\tilde{\delta}$. 
 \end{enumerate}
\end{theorem}

Algorithm \textsf{Construct-coordinate-oracle}
gives the algorithm for constructing coordinate oracles. 

\begin{algorithm}[H]
    \SetKwFor{RepeatTimes}{repeat}{times}{end}
    \KwIn{$f$ (target function), 
   $k$ (arity of Junta),  \\ $\delta$ (confidence parameter),  $\nu \le 1/8$ (first accuracy parameter), $\tau$ (second accuracy parameter)}
    \KwOut{an oracle $\dictset$} 
    \tcp{Construct the initial oracles}
    Let $T = C k^5 \tau^{-5} \red{\log(1/\delta)}$ and let $M = C k^7 \tau^{-7} \red{\log(1/\delta)}$\;
    Let $\tilde \delta = \delta/(MT)$ \;
    Initialize $\dictset = \emptyset$ \;
    \RepeatTimes{$T$ \label{line:coord-oracle-first-loop}}{
        Sample $\rho$ according to $\mathcal{R}_{1/k}$ (as in Definition~\ref{def:restriction})\;
        \RepeatTimes{$M$}{
            Sample $\randeta$ as in Lemma~\ref{lem:isolate-variable}\;
            Let $g(x) = \sgn(\Has{\randeta} f_{\res \rho}(x) - \E[f_{\res \rho}]^3)$\;
            Apply \textsf{Dictator-test} to $g$ with confidence $\tilde \delta$ and accuracy $\nu$\;
            \If{\textsf{Dictator-test} accepts}{
                Add $g$ to $\dictset$ \;
            }
        }
    }

    \tcp{Clean out duplicates}
    Let $N = C \log(MT/\delta)$, and sample $\bx^{(1)}\,\dots, \bx^{(N)} \in \{\pm 1\}^n$ independently and uniformly \;
    \While{there exist $g\ne h \in \dictset$ such that $|N^{-1} \sum_i g(\bx^{(i)}) h(\bx^{(i)})| \ge \frac 12$\label{line:coord-oracle-second-loop}}{
        Remove $g$ from $\dictset$\;
    }
    \Return{$\dictset$}
    \caption{\textsf{Construct-coordinate-oracle}}\label{alg:coord-oracle}
\end{algorithm}

\begin{lemma}\label{lem:coord-oracle}
The algorithm~\textsf{Construct-coordinate-oracle} has the following guarantee: 
\begin{enumerate}
\item As input, it gets oracle access to $f: \{-1,1\}^n \rightarrow \{-1,1\}$, an arity parameter $k$, two accuracy parameters $\tau$ and $\nu \le 1/8$ and a confidence parameter $\delta>0$. 
\item With probability $1-\delta$, there is some set $\coords \supseteq \{i: \Inf_i^{\le k}(f) \ge \tau^2/k^2\}$
    such that the output of Algorithm~\ref{alg:coord-oracle} is an $\nu$-oracle
    to $\coords$. 
    \item The number of oracles in $\dictset$ is at most $\poly(k, \tau^{-1}, \log(1/\delta))$. 
    \item The query complexity of the procedure is $\nu^{-1} \cdot \poly(k, \tau^{-1}, \log(1/\delta))$. 
\end{enumerate}  
\end{lemma}
\ignore{
\begin{lemma}\label{lem:coord-oracle}
    With probability
    at least $1-\delta$, there is some set $\coords \supseteq \{i: \Inf_i^{\le k}(f) \ge \tau^2/k^2\}$
    such that the output of Algorithm~\ref{alg:coord-oracle} is an $\epsilon$-oracle
    to $\coords$. \red{Further, the size of the output is $\poly(k, \tau^{-1}, \log(1/\delta))$.}
\end{lemma}}

\begin{proof}
{\red{Note that the number of oracles in $\dictset$ is at most $MT$. This immediately gives Item~3. Similarly, the query complexity of the algorithm is $M \cdot T$ times the cost of applying the \textsf{Dictator-test} in Step~9. By plugging the bound from~Theorem~\ref{thm:dictatorship-test}, we get Item~4. This leaves us with  Item~2. 
Like the algorithm itself, there are two steps in this proof.}}
The first step is to show that after executing the loop on line~\ref{line:coord-oracle-first-loop},
the set $\dictset$ contains only $\nu$-approximate dictators, and it contains at least
one dictator for every coordinate $i$ with large $\Inf_i^{\ge k}(f)$.
The second step is to show that by the end of the algorithm, each coordinate that was represented
by a dictator after executing the loop on line~\ref{line:coord-oracle-first-loop} is now
represented by a unique dictator. Actually, this second step is trivial: 
if $g$ and $h$ are $\nu$-approximate dictators for $\nu \le \frac 18$ then
$|\E[gh]| \le \frac 14$ if they represent different coordinates, while $|\E[gh]| \ge \frac 34$ if they represent the
same coordinate.
By a union bound over the $O(M^2 T^2)$ pairs of elements in $\dictset$, $\log(C M T/\delta)$ samples
are enough to estimate all of these correlations to accuracy $\frac 14$ with confidence $1-\delta/4$,
meaning that the loop starting on line~\ref{line:coord-oracle-second-loop} correctly (with probability at least $1-\delta/4$)
chooses exactly one dictator function to represent each coordinate.

We turn to the correctness of the first loop: it should find a dictator for every influential coordinate.
The analysis of this part is itself divided into two parts, corresponding to the two nested loops:
we will argue that for every influential coordinate $i$, at least one iteration of the outer
loop will sample $\rho$ for which $\widehat{f_{\res \rho}}(i)$ is large. For this execution of the outer loop,
we will argue that at least one iteration of the inner loop will find an oracle for coordinate $i$.

We will write the first of these parts as a separate claim, and prove it later:
\begin{claim}\label{clm:restrict}
    If $\Inf_i^{\le k}(f) \ge \tau^2/k^2$ and $\rho \sim \mathcal{R}_{1/k}$ then with
    probability $\Omega(k^{-4} \tau^4)$, $|\widehat{f_{\res \rho}}(i)| \ge \tau/(4k)$.
\end{claim}

With our choice of $T$, Claim~\ref{clm:restrict} and a union bound imply that with probability
at least $1 - \delta/3$, for every $i$ with $\Inf_i^{\le k}(f) \ge \tau^2/k^2$ there is at least one
iteration of the outer loop for which $|\widehat{f_{\res \rho}}(i)| \ge \frac{\tau}{4k}$.
We will fix this iteration, and examine the inner loop: by Lemma~\ref{lem:isolate-variable}
with $\kappa = \frac{\tau}{4k}$, each iteration has $\Omega(\tau^6 k^{-6})$ probability 
of producing $g$ that is equal to $\pm \Dict_i$. By a union bound, with probability at least
$1 - \delta/3$, the inner loop will succeed in adding $\pm \Dict_i$ to $\dictset$.

Finally, a union bound implies that with probability at least $1 - \delta/3$, every $g$ that
passes the dictatorship test is in fact $\epsilon$-close to some dictator.
\end{proof}

\begin{proof}[Proof of Claim~\ref{clm:restrict}]
    For $i \in [n]$, define the polynomial $D_i f: \R^n \to \R$ by
    \[
        (D_i f)(x) = \frac{1}{2} (f(x) - f(x_{-,i})),
    \]
    where $x_{-,i}$ is equal to $x$, except that the $i$th coordinate is negated. In terms of Fourier coefficients,
    it is easy to check that
    \[
        (D_i f)(x) = x_i^2 \sum_{S \subseteq [n] \setminus \{i\}} \hat f(S \cup \{i\}) \chi_S(x).
    \]
    On the other hand, it is also immediate that for every $x \in \{-1, 0, 1\}^n$,
    and it is easy to check that $\widehat{f_{\res \rho}}(i) = D_i f(\rho)$. In particular,
    \begin{align*}
        \Var_\rho[\widehat{f_{\res \rho}}(i)]
        &= \Var_\rho[D_i f(\rho)] \\
        &= \sum_{S \subseteq [n] \setminus \{i\}} \hat f^2(S \cup \{i\}) (1 - 1/k)^{|S|} \\
        &\ge (1-1/k)^k \sum_{\substack{S \subseteq [n] \setminus \{i\} \\ |S| \le k}} \hat f^2(S \cup \{i\}) \\
        &\ge \frac 14 \Inf_i^{\le k}(f) \\
        &\ge \frac{\eta^2}{4k^2}.
    \end{align*}

    We will apply an anti-concentration inequality to the random variable $D_i f(\rho) = \widehat{f_{\res \rho}}(i)$
    (see Lemma~13 in~\cite{De:2018BFA}):
    for any real-valued random variable $\bX$ with variance at least $\sigma^2$ and central fourth moment
    at most $t^4 \sigma^4$,
    \[
        \Pr\left[|X| > \frac{\sigma}{2} \right] \ge \frac{9}{128(t + 2)^4}.
    \]
    Immediately from the definition of $D_i f$, we see that the central fourth moment of $D_i f(\rho)$
    is at most $16$. Setting $\sigma = \eta/(2k)$ and $t = 4k/\eta$, we have
    \[
        \Pr\left[
            |\widehat{f_{\res \rho}}(i)| > \frac{\eta}{4k}
        \right]
        \ge \Omega(k^{-4} \eta^{4}),
    \]
    as claimed.
\end{proof}
{\red{Finally, the following claim, which will be useful in both Section~\ref{sec:exp-query} and Section~\ref{sec:poly-tester}, says that if $f$ correlates with 
a $k$-junta, then it also correlates nearly as well with a junta on some set $S$ such that for every variable $j \in S$, $\mathsf{Inf}_j^{\le k} (f)$ is large. This is useful because Lemma~\ref{lem:coord-oracle} can then be used to construct oracles to all these \emph{relevant variables}.}}
\begin{claim}~\label{clm:correlation1}
 Let $f: \{\pm 1\}^n \rightarrow \{\pm 1\}$ and let $g: \{\pm 1\}^n \rightarrow \{\pm 1\}$ be a $k$-junta 
on the variables $\{1,\ldots, k\}$  
 such that $\mathbf{E}[f \cdot g] \ge c$. Then, for any $\tau>0$, there is a set $S \subseteq [k]$ of variables such that 
 \begin{enumerate}
 \item For $i \in S$, $\mathsf{Inf}_i^{\le k}(f) \ge \frac{\tau^2}{k^2}$. 
 \item There is a junta on $S$ with correlation $c-\tau$ with $f$. 
 \end{enumerate}
 \end{claim}
\begin{proof}
Let us first start with $S = [k]$. If all variables $i \in S$ satisfy $\mathsf{Inf}_i^{\le k}(f) \ge \frac{\tau^2}{k^2}$, we are done. Else, 
let $j$ be any variable such that $\mathsf{Inf}_j^{\le k}(f) \le \frac{\tau^2}{k^2}$. For this variable $j$, define $g_j: \{\pm 1\}^n \rightarrow [-1,1]$ as the junta on the set $[k] \setminus \{j\}$ 
\[
g_j(x) = \frac12 \big( g(x) + g(x^{\oplus j})\big),
\]
where $x^{\oplus j}$ is the same as $x$ with the coordinate at $j$ flipped.  Observe that for any $S$, 
$\widehat{g}_j(S) =\widehat{g}(S)$ if $j \not \in S$ and $0$ otherwise. 
Note that $$ 
|\mathbf{E}_{x} [f\cdot g_j(x)] -  \mathbf{E}_{x} [f\cdot g(x) ]|  = | \sum_{S  \ni j} \widehat{g}(S) \widehat{f}(S) | = | \sum_{S  \ni j : |S| \leq k} \widehat{g}(S) \widehat{f}(S) | \le \sqrt{\sum_{S \ni j: |S| \le k} \widehat{f}^2(S) } \le \sqrt{\Inf_j^{\le k} (f)}. 
$$ 
Thus, $\mathbf{E}[f\cdot g_j(x)]] \ge c -\frac{\tau}{k}$. By doing a simple randomized rounding step, we can in fact, assume that the range of $g_j$ is $\pm 1$.  We now set $S \leftarrow [k] \setminus j$ and $g= g_j$ and inductively repeat the argument. This finishes the proof. 
\end{proof}

\section{Description of \textsf{Maximum-correlation-junta}}~\label{sec:exp-query}
In this section, we want to prove Theorem~\ref{thm:main-tester}. 
The final ingredient required to describe the algorithm \textsf{Maximum-correlation-junta} (from Theorem~\ref{thm:main-tester}) is the routine \textsf{Find-best-fit}. Here, we state the algorithmic guarantee for this routine.  The description and its proof of correctness is deferred to Appendix~\ref{sec:best-fit}. 
\begin{restatable}{lemma}{rl}~\label{lem:find-best-fit}
There is an algorithm \textsf{Find-best-fit}
with the following guarantee: 
\begin{enumerate}
\item The algorithm gets as input oracle access to a function $f: \{-1,1\}^n \rightarrow \{-1,1\}$ as well as a set of oracles $\mathcal{S} \subseteq \{\mathsf{Dict}_1, \ldots, \mathsf{Dict}_n\}$. We clarify that algorithm is only given the oracle $\mathsf{Dict}_i$ (for $\mathsf{Dict}_i \in \mathcal{S}$) but not $i$. 
\item The algorithm gets as input error parameter $\epsilon>0$, arity parameter $k$ and confidence parameter $\delta>0$. 
\item The algorithm  makes $N(k,|\mathcal{S}|, \epsilon, \delta) = O(2^k/\epsilon^2 \cdot |\mathcal{S}| \cdot (\log (1/\delta) + k^2 + |\mathcal{S}|))$ queries with probability $1-\delta$. 
\item Each query point (to either $f$ or oracle in $\mathcal{S}$) is distributed as a uniformly random element of $\{-1,1\}^n$. 
\item With probability $1-\delta$, the algorithm outputs a number $\widehat{\mathsf{Corr}}_{f,\mathcal{S},k}$ and $h_{\mathcal{S},k}:\{-1,1\}^k \rightarrow \{-1,1\}$ such that there exists oracles $\mathsf{Dict}_{i_1}, \ldots, \mathsf{Dict}_{i_k} \in \mathcal{S}$ 
\[
 \mathbf{E}_{\bx} [f(\bx)  \cdot h_{\mathcal{S},k}(\bx_{i_1}, \ldots, \bx_{i_k})] \ge \max_{\ell \in \mathcal{J}_{\mathcal{S},k}}\mathbf{E}_{\bx} [f(\bx)  \cdot \ell(\bx)] - \epsilon \ \ \textrm{and} \ \ |\widehat{\mathsf{Corr}}_{f,\mathcal{S},k} - \max_{\ell \in \mathcal{J}_{\mathcal{S},k}}\mathbf{E}_{\bx} [f(\bx)  \cdot \ell(\bx)]| \le \epsilon.
\]
\end{enumerate} 
\end{restatable}
We next describe the algorithm \textsf{Maximum-correlation-junta} (Algorithm~\ref{alg:maximum-junta}).

\begin{algorithm}[H]
    \SetKwFor{RepeatTimes}{repeat}{times}{end}
    \KwIn{$f$ (target function), 
   $k$ (arity of Junta),    $\epsilon$ (distance parameter)}
    \KwOut{$\widehat{\mathsf{Corr}}_{f,k} \in [0,1]$ and $h: \{-1,1\}^k \rightarrow \{-1,1\}$} 
    \tcp{Set parameters for the algorithm}
    Let $\tau=\frac{\epsilon}{4}$, $\delta =1/3$ \; 
    Let $\mathsf{R} = \poly(k,\tau^{-1}, \log(1/\delta))$  -- the upper bound 
    on output size from Item~3 of Lemma~\ref{lem:coord-oracle} \; 
    Set $N = N(k, R, \epsilon/4, \delta/4)$ where $N(\cdot)$ is the function in Item~3 of Lemma~\ref{lem:find-best-fit} \;
    Set $\nu = \delta/ 4N$ \;
    \tcp{Construct the initial oracles}
Run \textsf{Construct-coordinate-oracle}   with input function $f$, junta arity parameter $k$,   confidence parameter $\delta/4$, first accuracy parameter $\nu$ and second accuracy parameter $\tau$ \;
    Let $\dictset$ be the set of returned oracles \; 
    \tcp{Use the oracles to find maximally correlated junta}
    Error parameter for queries to any oracle $g\in \dictset$ is set to $\nu$ \;
    Run \textsf{Find=best-fit} with target function $f$, oracles given by $\dictset$, confidence parameter $\delta/4$, distance parameter $\epsilon/4$ and arity parameter $k$. \;
    Let the output of this routine be $ \widehat{\mathsf{Corr}}_{f,\mathcal{D},k}$ and  $h_{\mathcal{D},k}: \{-1,1\}^k \rightarrow \{-1,1\}$. \; 
    Set $h \leftarrow h_{\mathcal{D},k}$ and $ \widehat{\mathsf{Corr}}_{f,k} \leftarrow \widehat{\mathsf{Corr}}_{f,\mathcal{D},k}$ \;
    \Return{($\widehat{\mathsf{Corr}}_{f,k}$, $h$)}
    \caption{\textsf{Maximum-correlation-junta}}\label{alg:maximum-junta}
\end{algorithm}
\begin{proofof}{Theorem~\ref{thm:main-tester}}
We begin with the following claim. 
\begin{claim}~\label{clm:upper-bound-sample}
The query complexity of the algorithm is $2^k \cdot \mathsf{poly}(k,1/\epsilon)$. 
\end{claim}
\begin{proof}
By just plugging the bounds, we see that $\mathsf{R}$ defined in Step~2 of the algorithm is $\poly(k/\epsilon)$; $N$ defined in Step~3 is $2^k \cdot \poly(k/\epsilon)$ and $\nu$ defined in Step~4 is $2^{-k} \cdot \poly(\epsilon/k)$. 

Now, the query complexity of Step~5, i.e., \textsf{Construct-coordinate-oracle} is then $\nu^{-1} \cdot \poly(k, \tau^{-1}) = 2^{k} \cdot \poly(k/\epsilon)$ (Item~4 of Lemma~\ref{lem:coord-oracle}).
Next, note that $|\dictset| \le \poly(k/\epsilon)$ (by plugging the bound from Item~3 of Lemma~\ref{lem:coord-oracle}). However, this means that the algorithm \textsf{Find-best-fit} (from Step~3 of Lemma~\ref{lem:find-best-fit}) makes at most  $2^k \cdot  \poly(k, |\dictset|) = 2^k \cdot \poly(k,1/\epsilon)$ queries where each query is either to $f$ or to an oracle in $\dictset$. However, each call to $\mathcal{D}$ is made with error $\nu$, the query complexity of making each such call is $\poly(k, \log(1/\nu)) = \poly(k)$. Thus, the total query complexity is $2^k \cdot \poly(k/\epsilon)$. 
\end{proof}
Having proven a bound on the query complexity, we now
turn to the proof of correctness of this algorithm. Note that every oracle $g \in \mathcal{D}$ is $\nu$ close to $\mathsf{Dict}_i$ for some $i \in [n]$. Further, at point $x$, (by definition of a $\nu$-oracle), we have an algorithm, which returns the value $g(x)$ with probability at least $1-\nu$. We say that the evaluation of $g$ at point $x$ is \emph{good}, if we get the value of $g(x) = \mathsf{Dict}_i(x)$. Note that a randomly chosen point $x \in \{-1,1\}^n$, 
\begin{eqnarray*}
\Pr_{\bx \in \{-1,1\}^n}[ \textrm{Evaluation of }g \textrm{ is not good}] &\le& \Pr_{\bx \in \{-1,1\}^n}[g(x) \not= \mathsf{Dict}_i(x)] + \Pr[\textrm{Evaluation of }g \textrm{ is incorrect}] \\ &\le& \nu + \nu = 2\nu. 
\end{eqnarray*}
Since the query by the algorithm \textsf{Find-best-fit}   to each oracle in $\dictset$ is a random point in $\{-1,1\}^n$ (Item~4 in Lemma~\ref{lem:find-best-fit}) and the total number of queries to \textsf{Find-best-fit} is $N$, hence the probability the evaluation of $g \in \dictset$ is not good at any point is at most $2N \nu =\delta/2$. Thus, from now on, we will assume that $g \in \dictset$ is $\nu$-close to $\mathsf{Dict}_i$, then the algorithm \textsf{Find-best-fit} has exact access to $\mathsf{Dict}_i$. However, with this the following claim is immediate from Item~5 of Lemma~\ref{lem:find-best-fit}. 
\begin{claim}\label{clm:soundness}
Suppose the algorithm outputs $(\widehat{\mathsf{Corr}_{f,k}},h)$. Let $\mathsf{V} = \{i \in [n]: \mathsf{Dict}_i$ is $\nu$-close to some $g \in \mathcal{D}\}$. Then, with probability $1-\delta/4$, there is some $\{i_1, \ldots, i_k\} \in \mathsf{V}$ such that 
\[
\mathbf{E}_{\bx} [f(\bx) \cdot h(\bx_{i_1}, \ldots, \bx_{i_k})] \ge \max_{\ell \in \mathcal{J}_{\mathsf{V},k}} \mathbf{E}_{\bx} [f( \bx) \cdot \ell(\bx)] - \frac{\epsilon}{4} \textrm{ and } |\widehat{\mathsf{Corr}_{f,k}} - \max_{\ell \in \mathcal{J}_{\mathsf{V},k}} \mathbf{E}_{\bx} [f( \bx) \cdot \ell(\bx)]| \le \frac{\epsilon}{4}. 
\] 
\end{claim}
Finally, we have the following claim. 
\begin{claim}~\label{clm:completeness}
Let $f:\{-1,1\}^n \rightarrow \{-1,1\}$ and assume that there exists a $k$-junta $g:\{-1,1\}^n \rightarrow \{-1,1\}$ such that $\mathbf{E}_{\bx} [f(\bx) \cdot g(\bx)] \ge \mathsf{Corr}_{f,g}$. Let $\mathcal{D}$ be the set of returned oracles in Step~6.  Let $\mathsf{V} = \{i \in [n]: \mathsf{Dict}_i$ is $\nu$-close to some $g \in \mathcal{D}\}$. Then, with probability $1-\delta/4$, there exists $T \subseteq \mathsf{V}$, $|T| =k$ and a junta on $T$ with correlation at least $\mathsf{Corr}_{f,g} - \epsilon/4$ 
with $f$. 
\end{claim}
\begin{proof}
Let $W$ be the set of variables appearing $g$. By Claim~\ref{clm:correlation1} (setting $\tau = \epsilon/4$), there is subset $W' \subseteq W$ and a $W'$-junta $r: \{-1,1\}^n \rightarrow \{-1,1\}$ with the following properties: (a) $\mathbf{E}_{\bx}[r(\bx) \cdot f(\bx) ] \ge \mathsf{Corr}_{f,g} - \frac{\epsilon}{4}$. (b) For every $i \in W'$, $\mathsf{Inf}^{\le k}_i(f) \ge \frac{\epsilon^2}{16k^2}$. Now, with the same value of $\tau$, applying Lemma~\ref{lem:coord-oracle} implies that $W' \subseteq \mathsf{V}$ with probability $1-\delta/4$. This finishes the claim. 
\end{proof}
The proof of correctness is now immediate from Claim~\ref{clm:completeness} and Claim~\ref{clm:soundness}. 
\end{proofof}

\section{The polynomial-query gap tester~\label{sec:poly-tester} }

Recall the context of the relaxed tester compared to the original one: we have
already identified (using Lemma~\ref{lem:coord-oracle}) a subset $\coords$ of
size $\poly(k, 1/\epsilon)$ such that all potentially interesting coordinates
of $f$ are contained in $\coords$ (in the sense of
Claim~\ref{clm:correlation1}).
Recall, moreover, that we do not have an explicit representation of the coordinates $\coords$;
we only have access to coordinate oracles in the sense of Definition~\ref{def:oracle}.
Nevertheless, let us first
pretend that we do have explicit access to $\coords$, in order to explain roughly
how our algorithm works.

The first step of our algorithm is to replace $f$ by $\favg$, defined by
$\favg(x) = \E_{\by \sim \Unif} [f(\by) \mid \by_\coords = x_\coords]$.
This step turns out to be unnecessary if we have explicit access to $\coords$,
but on the other hand it is also clearly harmless, because $\favg$ is the
orthogonal projection of $f$ onto the space of $\coords$-juntas, and hence
$\E[fg] = \E[\favg g]$ for every $g \in \junta_{\coords,k'}$.

The second step of our algorithm is to replace $\favg$ with $\fsmooth$, defined
by $\fsmooth = T_{1-\epsilon/(2k)} \favg$.
Applying this noise has two nice features: it approximately preserves correlation
with $k$-Juntas, and it allows us to bound total influence. First, we will show
that it preserves correlation:

\begin{lemma}\label{lem:smoothing}
    For every $k$-Junta $g$, if $s = \epsilon/(2k)$ then
    \[
        \|g - T_{1-s} g\|_2 \le \frac{\epsilon}{2} \|g\|_2.
    \]
    In particular, for all $f: \{\pm 1\}^n \to [-1, 1]$ and $g \in \junta_{\coords,k}$.
    \[
        \left|\E [(T_{1-s} f) g] - \E [fg]\right| \le \epsilon/2.
    \]
\end{lemma}

Consequently,  $\max_{g \in \junta_{\coords,k}} \E[\fsmooth g]$
is essentially the same as $\max_{g \in \junta_{\coords,k}} \E[\favg g]$.

\begin{proof}
    If $g$ is a $k$-Junta then $\hat g_S = 0$ whenever $|S| > k$, and so
    \[
        \|g - T_{1-s} g\|_2^2 = \sum_{|S| \le k} (1 - (1-s)^{|S|})^2 \hat g^2(S) \le \frac{\epsilon^2}{4} \sum_{S} \hat g^2(S)
        \le \frac{\epsilon^2}{4} \|g\|_2^2.
    \]
    The second claim follows from the fact that $\E[(T_{1-s} f) g] = \E [f (T_{1-s} g)]$, and so $\E[(T_{1-s} f - f) g] = \E[f (T_{1-s} g - g)]$.
    Then apply Cauchy-Schwarz and the first claim (together with the fact that $\|f\|_2$ and $\|g\|_2$ are at most 1).
\end{proof}

As we claimed above, we can also bound the total influence of a smoothed function:

\begin{lemma}\label{lem:high-influence-coords}
    If $f: \{\pm 1\}^n \to [-1, 1]$ then $\Inf(T_{1-s} f) \le \frac{1}{2 s}$. In particular,
    \[
        |\{1 \le j \le n: \Inf_j(T_{1-s} f) \ge 2s\}| \le \frac{1}{4s^2}.
    \]
\end{lemma}

\begin{proof}
    Since $f$ takes values in $[-1, 1]$, $\E[f^2] \le 1$ and so $\sum_S \hat f^2(S) \le 1$. On the other hand,
    \[
    \Inf(T_{1-s} f) = \sum_S |S| (1 - s)^{|S|} \hat f^2(S) \le \max\{j (1-s)^j: 1 \le j \le n\} \le \frac{1}{e s} \le \frac{1}{2 s}
    \]
    This proves the first claim. The second follows from Markov's inequality.
\end{proof}

According to Lemma~\ref{lem:high-influence-coords}, if $\coords' = \{i \in \coords: \Inf_i(\fsmooth) \ge \epsilon/k\}$ then
$|\coords'| \le k^2/\epsilon^2 = k'$.
On the other hand, we can easily see that coordinates not belonging to $\coords'$ are irrelevant when
it comes to approximating $\fsmooth$ by a $k$-Junta:

\begin{lemma}\label{lem:low-influence-coords}
    For $f: \{\pm 1\}^n \to [-1, 1]$, let $\coords' \subset \coords \subset [n]$ satisfy $i \in \coords'$ whenever
    $\Inf_i(f) \ge t$. Then
    \[
        \left|\max_{g \in \junta_{\coords,k}} \E[f g] - \max_{g \in \junta_{\coords',k}} \E[f g] \right| \le tk.
    \]
\end{lemma}

\begin{proof}
    Let $g \in \junta_{\coords,k}$ maximize $\E[f g]$, and let $\othercoords$ be the set
    of coordinates on which $g$ depends. Define $\gavg$ by $\gavg(x) = \E_{\by \sim \Unif}[g(\by) \mid \by_S = x_S]$.
    Then $\gavg$ is a $(\coords' \cap \othercoords)$-Junta; it does not
    necessarily belong to $\junta_{\coords',k}$ because it is not necessarily
    Boolean.  However,  $\max_{h \in \junta_{\coords',k}} \E[f h] \ge \E[f\gavg]$ -- this is because
    $\gavg$ can be rounded to a Boolean function without decreasing $\E[f \gavg]$.
    Therefore, it suffices to show that
    \[
        \E[f\gavg] \ge \E[f g] - tk.
    \]
    To do this, note that the Fourier coefficients of $g$ and $\gavg$ are related by
    $\hat g_S = (\widehat \gavg)_S$ whenever $S \subset \coords'$, and $(\widehat \gavg)_S = 0$ otherwise.
    In particular,
    \[
        \E[f (g - \gavg)] = \sum_{\substack{S \subset \othercoords \\ S \not \subset \coords'}} \hat g_S \hat f_S
        \le \sum_{\substack{S \subset \othercoords \\ S \not \subset \coords'}} \hat f_S^2
        \le \sum_{i \in \othercoords \setminus \coords'} \Inf_i(f) \le tk,
    \]
    where the last inequality follows because $|\othercoords \setminus \coords'| \le k$.
\end{proof}

Thanks to Lemma~\ref{lem:low-influence-coords}, if we set $\coords' = \{i \in \coords: \Inf_i(\fsmooth) \ge \epsilon/k\}$
then $\max_{g \in \junta_{\coords,k}} \E[\fsmooth g] \approxeq \max_{g \in \junta_{\coords',k}} \E[\fsmooth g]$.
We do not know how to estimate this final quantity, but we can easily estimate (the larger quantity)
$\max_{g \in \junta_{\coords',k'}} \E[\fsmooth g]$
(note that $k$ has become $k'$). This is because $\junta_{\coords',k'}$ consists of all boolean functions
depending on the coordinates in $\coords'$. The one with maximal correlation can be found
by projecting $\fsmooth$ onto the space of $\coords'$-Juntas and rounding the result to a boolean function:
define $\fsmoothavg$ by $\fsmoothavg(x) = \E[\fsmooth(\by) \mid \by_{\coords'} = x_{\coords'}]$. Then
\[
 \max_{g \in \junta_{\coords',k'}} \E[\fsmooth g] = \E[|\fsmoothavg|].
\]
This number (or rather, an estimate of it) will be the final output of our algorithm.
Thanks to the preceding arguments, it is larger (or at least, not much smaller) than
$\max_{g \in \junta_{n,k}} \E[fg]$. On the other hand, it is certainly
smaller than $\max_{g \in \junta_{n,k'}} \E[fg]$.

In order to turn the description above into an algorithm, we need to describe how to
compute all the quantities above given implicit access to the coordinates in $\coords$.
In particular, we will first describe how to simulate query access to $\favg$.
Given this, it is obvious how to simulate query access to $\fsmooth$.
Then, we will show how to estimate $\coords'$; more accurately, we will
compute some $\tilde \coords'$ that contains all coordinates of influence at least $2\epsilon/2k$
and no coordinates of influence smaller than $\epsilon/k$.
Finally, we will show how to simulate query access to $f_{\mathsf{smooth},\mathsf{avg},\tilde \coords'}$;
we can use this query access to estimate $\E[|f_{\mathsf{smooth},\mathsf{avg},\tilde \coords'}|]$, which
is the final output of our algorithm.

\subsection{Averaging over irrelevant coordinates}\label{sec:averaging}

An important primitive for us will be the ability to average over irrelevant
coordinates given oracle access to the relevant ones.  That is, imagine that
there is a collection of coordinates $\coords \subset [n]$.  Given query access
to $f: \{\pm 1\}^n \to [-1, 1]$ and a fixed point $x \in \{\pm 1\}^\coords$, we
would like to estimate $\E_{\by \sim \Unif} [f(\by) \mid \by_\coords = x]$.  Were we given
$\coords$ explicitly, this would be easy; our challenge is to do it with only
oracle access to $\coords$, in the sense of Definition~\ref{def:oracle}.

\begin{algorithm}[H]
    \KwIn{a function $f$, $x \in \{\pm 1\}^n$, an oracle $\dictset$, $\gamma$ (accuracy parameter), $\delta$ (failure parameter)}
    \KwOut{a number}

    Let $T = C |\dictset| \frac{\log 1/\delta}{\gamma^2}$ \;
    Let $\bx^{(1)} = x$ \;
    \For{$i = 1$ to $T - 1$}{
        \Repeat{$g(\by) = g(x)$ for all $g \in \dictset$ \label{line:coord-test}} {
            Let $\by$ be a copy of $\bx^{(i)}$, but flip each bit independently with probability $\frac{1}{2|\dictset|}$\;
        }
        Let $\bx^{(i+1)} = \by$ \;
    }
    \Return{$\frac 1T \sum_{i=1}^T f(\bx^{(i)})$}
    \caption{\textsf{Coordinate-projection}\label{alg:averaging}}
\end{algorithm}

\begin{lemma}\label{lem:averaging}
    \textsf{Coordinate-projection} has the following guarantees.
    Given a function $f: \{\pm 1\}^n \to [-1, 1]$, a point $x \in \{\pm 1\}^n$,
    an oracle $\dictset$ for the coordinates $\coords \subset [n]$,
    and parameters $\delta, \gamma > 0$, the algorithm makes (in expectation)
    $\Theta(|\coords| \frac{\log 1/\delta}{\gamma^2})$ queries to $f$
    and to each element of $\dictset$ and, with probability at least $1 - \delta$, outputs
    a number within $\gamma$ of $\E_{\by \sim \Unif} [f(\by) \mid \by_\coords = x]$.
\end{lemma}

\begin{remark}\label{rem:averaging}
    It is not hard to see that \textsf{Coordinate-projection} succeeds even if it is given noisy
    query access to $f$. For example, if each query that
    \textsf{Coordinate-projection} makes to the function $f$ can produce an error of at most $\gamma$,
    then (under the assumptions of Lemma~\ref{lem:averaging}), the output of \textsf{Coordinate-projection}
    is, with probability at least $1 - \delta$, accurate to within $2\gamma$.
\end{remark}

\begin{proof}
    First, observe that for every $i = 1, \dots, T$, $\bx^{(i)}_\coords = x_\coords$:
    indeed, $\dictset$ is an oracle for $\coords$ and so the test on line~\ref{line:coord-test}
    will only pass if $\by_\coords = x_\coords$. Next, observe that conditioned on $\bx^{(i)}$,
    every coordinate $j \not \in \coords$ of $\bx^{(i+1)}$ is obtained by (independently)
    flipping $\bx^{(i)}_j$ with probability $\frac{1}{2|\coords|}$. In other
    words, $\bx^{(1)}_{\bar \coords}, \dots, \bx^{(T)}_{\bar \coords}$ is a
    Bonami-Beckner Markov chain with correlation parameter $\rho = 1 -
    \frac{1}{|\coords|}$.
    If $\bz$ is distributed according to the stationary distribution
    of this Markov chain then $\bz_\coords = x_\coords$ and $\bz_{\bar \coords}$
    is uniformly distributed on $\{\pm 1\}^{\bar \coords}$. In particular,
    $\E[f(\bz)] = \E_{\by \sim \Unif}[f(\by) \mid \by_\coords = x_\coords]$ and so Lemma~\ref{lem:bias}
    implies that with probability at least $1-\delta$, $\frac{1}{T} \sum_{i=1}^T f(\bx^{(i)})$
    is within $\gamma$ of $\E_\by[f(\by) \mid \by_\coords = x_\coords])$.

    It remains to check how many oracle queries our procedure makes; let $X_i$
    be the number of attempts it takes to successfully generate $\bx^{(i)}$.
    For each $i$, the probability that the test on line~\ref{line:coord-test}
    succeeds is exactly $(1 - \frac{1}{2|\coords|})^{|S|} \ge \frac 14$.
    In particular, each sample takes (in expectation) at most four queries to each element of $\dictset$,
    and so the overall number of queries is at most (in expectation) $O(T)$ to each element of $\dictset$.
\end{proof}

\subsection{Estimating influences}

One of the things we need to do is to estimate the influences of coordinates in
$\coords$. Again, if we had explicit access to these coordinates then this
would be trivial. The point is to get by with only oracle access, and the difficulty
is that given some $x \in \{\pm 1\}^n$, we cannot simply ``flip'' a bit belonging to $\coords$
because we don't know which bits those are. To work around this, we introduce the
notion of an influence-testing sample, which essentially is a collection of points
that manage to flip each bit in $\coords$:

\begin{definition}
    We say that $\infsamp \subset \{\pm 1\}^n$ is an \emph{influence-testing sample} at $x \in \{\pm 1\}^n$ with respect to $\coords$
    if we can enumerate $\infsamp = \{y^{(i)}: i \in \coords\}$ where $y^{(i)}_i = -x_i$ and $y^{(i)}_j = x_j$ for $j \in \coords \setminus \{i\}$.
\end{definition}

Note that the definition above doesn't guarantee anything about bits not belonging to $\coords$.
The point of the definition above is that if the function $f$ depends only on coordinates in $\coords$, then we can use
an influence-testing sample at a random point to obtain an unbiased estimator for all the influences of $f$.

The first important point about influence-testing samples is that we can produce them
given oracle access to $\coords$.

\begin{algorithm}[H]
    \KwIn{$x \in \{\pm 1\}^n$, and an oracle $\dictset$ for $\coords$}
    \KwOut{an influence-testing sample at $x$ with respect to $\coords$}

    Define (for $z \in \{\pm 1\}^n$) $I(z) = \{g \in \dictset: g(z) \ne g(x)\}$ \;
    Initialize $\infsamp = \emptyset$ \;
    \While{$|\infsamp| < |\dictset|$}{
        Let $\by$ be a copy of $x$, but flip each bit independently with probability $\frac{1}{2|\dictset|}$\;
        \If{$|I(\by)| = 1$ and $I(\by) \ne I(z)$ for all $z \in \infsamp$ \label{line:inf-samp-test}}{
            Add $\by$ to $\infsamp$ \;
        }
    }
    \Return{$\infsamp$}
    \caption{\textsf{Influence-testing-sample}}\label{alg:create-inf-samp}
\end{algorithm}

\begin{lemma}\label{lem:create-inf-samp}
    Let $\dictset$ be an oracle for $\coords$.
    Given $\dictset$ and $x \in \{\pm 1\}^n$ as input, \textsf{Influence-testing-sample}
    produces an influence-testing sample at $x$ with respect to $\coords$, while making
    (in expectation) $O(|\coords| \log |\coords|)$ queries to each element of $\dictset$.
\end{lemma}

\begin{proof}
    Each element that \textsf{Influence-testing-sample} adds to $\infsamp$ differs from $x$
    on exactly one coordinate of $\coords$. Moreover, the test on line~\ref{line:inf-samp-test} ensures
    that every coordinate of $\coords$ is represented by at most one element of $\infsamp$.
    Hence, by the time the loop is complete, $\infsamp$ is an influence-testing sample at $x$ with
    respect to $\coords$.

    Note that each time we sample $\by$, we have $\Pr[|I(\by)| = 1] = (1 - 2|\coords|^{-1})^{|\coords|} \ge \frac 14$.
    Moreover, conditioned on $|I(\by)| = 1$, $I(\by)$ is uniformly distributed among all size-one subsets of $\dictset$.
    Hence, the number of times that we need to sample $\by$ in order to see all size-one subsets at least once
    is distributed according to the coupon collector problem, and hence takes $O(|\coords| \log |\coords|)$ iterations
    in expectation.
\end{proof}

The other important point about influence-testing samples is that if we can use them to estimate influences.
The basic idea is to take the trivial algorithm for estimating influences (sample random elements, and check
whether flipping the $j$th bit changes the value of $f$), but using the notion of an influence-testing sample
to replace the need to bits; in \textsf{Threshold-influences}, $\bx^{(i)}$ are the random points
at which we're testing bit-flips, and $\by^{(i)}_g$ is the copy of $\bx^{(i)}$ with a bit (the one corresponding
to $g \in \dictset$) flipped.

\begin{algorithm}[H]
    \KwIn{a function $f$, an oracle $\dictset$ for $\coords$, threshold parameter $t$, failure parameter $\delta$}
    \KwOut{an oracle $\dictset'$}

    Let $T = C t^{-2} \log \frac{1}{\delta |\dictset|}$ \;
    For $s = 1, \dots, T$, sample $\bx^{(i)}$ uniformly from $\{\pm 1\}^n$ \;
    For $s = 1, \dots, T$, let $\{\by^{(i)}_g: g \in \dictset\}$ be an influence testing sample at $\bx^{(i)}$ (from \textsf{Influence-testing-sample}), where $\by^{(i)}_g$ is the element for which $g(\by^{(i)}) \ne g(\bx^{(i)})$ \;
    For $g \in \dictset$, let $\widehat{\Inf_g} = \frac 1T \sum_{i=1}^T (f(\bx^{(i)}) - f(\by^{(i)}_g))^2$ \;
    \Return{$\dictset' = \{g \in \dictset: \widehat{\Inf_g} \ge \frac 32 t\}$}
    \caption{\textsf{Threshold-influences}}\label{alg:threshold-influences}
\end{algorithm}

\begin{lemma}\label{lem:threshold-influences}
    Assume that $\dictset$ is an oracle for $\coords$ and that $f$ is a $\coords$-junta.
    With probability at least $1-\delta$, the output of \textsf{Threshold-influences}
    satisfies the following: there is a set $\coords' \subset \coords$ such that $\Inf_i(f) \ge 2t$ implies
    that $i \in \coords'$, and $i \in \coords'$ implies that $\Inf_i(f) \ge t$, and the output
    of \textsf{Threshold-influences} is an oracle for $\coords'$.

    Moreover, \textsf{Threshold-influences} makes $O(t^{-2} \log \frac{1}{\delta |\coords|})$ queries to
    $f$ and (in expectation) $O(t^{-2} |\coords|^2 \log \frac{1}{\delta})$ queries to each element of $\dictset$.
\end{lemma}

\begin{remark}\label{rem:threshold-influences}
    It is not hard to see that \textsf{Threshold-influences} succeeds even if it is given access
    to a slightly noisy version of $f$: if each evaluation
    of $f$ is guaranteed to be accurate to within $t/10$ then the guarantees of Lemma~\ref{lem:threshold-influences}
    still hold (at the cost of increasing the constant $C$ in \textsf{Threshold-influences}).
\end{remark}

\begin{proof}
    Let us fix a single $j \in \coords$ and its oracle $g \in \dictset$. We will show that
    with probability at least $1 - \frac{\delta}{|\coords|}$, if $\Inf_j(f) \le t$ then $g \in \dictset'$
    and if $\Inf_j(f) \ge 2t$ then $g \not \in \dictset'$ (and then the lemma will follow by a union bound over
    $g$). To prove this claim, it suffices to show that (with probability at least $1 - \frac{\delta}{|\coords|}$)
    $|\Inf_j(f) - \widehat \Inf_g| \le t/2$.

    For $z \in \{\pm 1\}^n$, let $\tilde z$ denote a copy of $z$ with bit $j$ flipped.
    Recall that $\Inf_j(f) = \E_{\bz \sim \Unif} [(f(\bz) - f(\tilde \bz))^2]$; by a Chernoff
    bound, with probability at least $1 - \frac{\delta}{|\coords|}$, $\Inf_j(f)$ is within $t/2$ of
    \[
        \frac{1}{T} \sum_{i=1}^T \left(f(\bx^{(i)}) - f(\tilde \bx^{(i)})\right)^2,
    \]
    where $\bx^{(i)}$ are (as in \textsf{Threshold-influences}) independent and uniform in $\{\pm 1\}^n$.
    Finally, recall that $f$ was assumed to be a $\coords$-Junta, and recall that among coordinates of $\coords$,
    $\by^{(i)}_g$ differs from $\bx^{(i)}$ exactly in coordinate $j$. Hence, $f(\by^{(i)}_g) = f(\tilde \bx^{(i)})$
    for all $i$, and so $\widehat \Inf_g$ coincides with the displayed quantity above.
    This proves the claim about the correctness of \textsf{Threshold-influences}; the claims
    about the number of queries follow immediately from the algorithm and Lemma~\ref{lem:create-inf-samp}.
\end{proof}

\subsection{The algorithm}

In order to set up the algorithm, recall from Section~\ref{sec:oracle}
that we can begin by finding an oracle $\dictset$ to a $\poly(k)$-sized set of relevant
coordinates $\coords$. As explained before, (e.g. by applying Claim~\ref{clm:correlation1}),
in order to complete the task it suffices
to compute $\max_{g \in \junta_{\coords,k'}} \E[fg]$ to accuracy $\epsilon$, where $k' = k^2/\eps^2$.

Recall that \textsf{Coordinate-projection} and \textsf{Threshold-influences}
have a failure probability $\delta > 0$; we will fix take $\gamma =
\poly(\epsilon, 1/k)$ and $\delta = 2^{-(k/\gamma)}$, and in what follows we will
guarantee to invoke these two algorithms at most $\poly(k, 1/\epsilon)$ times,
meaning that with high probability every single invocation will succeed.

Here is the final algorithm, together with a justification of the sample complexity:

\begin{itemize}
    \item $\favg: \{\pm 1\} \to [-1, 1]$ is defined by $\favg(x) = \E_{\by \sim
        \Unif}[f(\by) \mid \by_\coords = x_\coords]$.  According to
        Lemma~\ref{lem:averaging}, for every $x \in \{\pm 1\}^n$ we can
        approximate $\favg(x)$ to within $\gamma$ using $\poly(k, 1/\epsilon)$
        queries. Moreover (recalling our choice of $\delta$), as long as we
        repeat this process at most $\poly(k, 1/\epsilon)$ times, with high probability
        we will never
        fail to obtain accuracy $\gamma$.
        Also, $\favg$ is an $\coords$-Junta.
    \item $\fsmooth$ is defined by $\fsmooth = T_{1-\epsilon/(2k)} \favg$.
        By na\"ive sampling and a Chernoff bound, for every $x$ in $\{\pm 1\}^n$
        we can (with probability at least $1-\delta$)
        approximate $\fsmooth(x)$ to within $2\gamma$ using $\poly(1/\gamma, \log(1/\delta))$
        queries (each with accuracy $\gamma$) to $\favg$. Thanks to the previous point,
        this can be done using $\poly(k, 1/\epsilon)$ queries to $f$, and thanks to the
        choice of $\delta$ we can repeat this $\poly(k, 1/\epsilon)$ times without failure.
    \item Thanks to Lemma~\ref{lem:threshold-influences} and Remark~\ref{rem:threshold-influences}
        (with $t = \epsilon/k$), we can compute an oracle 
        $\tilde \dictset'$ to some set $\tilde \coords'$ such that $i \in \tilde \coords'$ whenever
        $\Inf_i(\fsmooth) \ge 2\epsilon/k$, and
        $\Inf_i(\fsmooth) \ge \epsilon/k$ for all $i \in \tilde \coords'$.
        This can be done using $\poly(k, 1/\epsilon)$ queries to $\fsmooth$, which (thanks to the previous point)
        can be done using $\poly(k, 1/\epsilon)$ queries to $f$.
        By Lemma~\ref{lem:high-influence-coords}, $|\tilde \coords'| \le k^2/\epsilon^2$.
    \item
        Applying Lemma~\ref{lem:averaging} again (this time, also applying Remark~\ref{rem:averaging}),
        for every $x \in \{\pm 1\}^n$ we can approximate
        $\ffinal(x) := \E_\by[\fsmooth(\by) \mid \by_{\tilde \coords'} = x_{\tilde \coords'}]$ to within $3\gamma$ using $\poly(k, 1/\epsilon)$ queries to $\fsmooth$ (which in turn requires $\poly(k, 1/\epsilon)$ queries to $f$).

    \item
        Finally, use $\poly(k, 1/\epsilon)$ independent samples to estimate
        $\E_\by [|\ffinal|]$ to within
        accuracy $4\gamma$ (which we can assume is at most $\epsilon$).
        This estimate is the output of the algorithm.
\end{itemize}

To complete the proof that this algorithm is correct (i.e.\ it fulfills the claims
made in Theorem~\ref{thm:main-tester-gap}), let us combine our previous bounds to show that
the final quantity being estimated in our algorithm (namely, $\E[|\ffinal|]$) is approximately
between $\max_{g \in \junta_{\coords,k'}} \E[fg]$ and $\max_{g \in \junta_{\coords,k}} \E[fg]$.

\begin{lemma}
    With $\ffinal$ defined as above,
    \[
        \max_{g \in \junta_{\coords,k'}} \E[fg] \ge \E[|\ffinal|] = \max_{g \in \junta_{n,k'}} \E[\ffinal g] \ge \max_{g \in \junta_{\coords,k}} \E[fg] - \frac 32 \epsilon.
    \]
\end{lemma}

\begin{proof}
    The first inequality follows immediately from Jensen's inequality and the fact
    that $\ffinal$ is defined by repeatedly averaging $f$ in various senses; we will focus on the last
    inequality.
    As we discussed previously,
    \[
        \max_{g \in \junta_{\coords,k}} \E[fg]
        = \max_{g \in \junta_{\coords,k}} \E[\favg g].
    \]
    By Lemma~\ref{lem:smoothing} applied to $\favg$,
    \[
        \max_{g \in \junta_{\coords,k}} \E[\fsmooth] \ge \max_{g \in \junta_{\coords,k}} \E[\favg g] - \epsilon/2.
    \]
    By Lemma~\ref{lem:low-influence-coords} applied with $\coords' = \tilde \coords'$ and $t = 2\epsilon/k$,
    \[
        \max_{g \in \junta_{\coords',k}} \E[\fsmooth g] \ge \max_{g \in \junta_{\coords,k}} \E[\fsmooth g] - \epsilon.
    \]
    Finally (and for the same reason as the first point),
    \[
        \E[|\ffinal|] = \max_{g \in \junta_{\coords',k'}} \E[\ffinal g] = \max_{g \in \junta_{\coords',k'}} \E[\fsmooth g].
    \]
\end{proof}

\bibliographystyle{plain}
\bibliography{allrefs}

\appendix
\section{Description of the algorithm \textsf{Find-best-fit}}~\label{sec:best-fit}
In this section, we describe the algorithm \textsf{Find-best-fit} (from Lemma~\ref{lem:find-best-fit}) and give its proof of correctness. We begin by describing the algorithm. 

\begin{algorithm}[H]
    \SetKwFor{RepeatTimes}{repeat}{times}{end}
    \KwIn{$f$ (target function), 
   $k$ (arity of Junta),    $\epsilon$ (distance parameter) \\ $\mathcal{S}$ (oracles to dictator functions), $\delta$ (confidence parameter)}
    \KwOut{$\widehat{\mathsf{Corr}}_{f,\mathcal{S},k} \in [0,1]$ and $h_{\mathcal{S},k}: \{-1,1\}^k \rightarrow \{-1,1\}$} 
    \tcp{Set parameters for the algorithm}
   Set $N =  2^k \cdot \frac{1}{\epsilon^2} \cdot 4 \cdot \big( |\mathcal{S}| + \log (1/\delta) + k^2\big)$ 
 and  $M = N \cdot 2^{-k}$  \; 
 \tcp{Evaluation of the oracles}
    Let $N' = \mathsf{Poi}(N)$ and sample 
$\bx^{(1)}, \ldots, \bx^{(N')} \in \{-1,1\}^n$ \; 
 Evaluate $f$ on $\bx^{(1)}, \ldots, \bx^{(N')}$ \;
 For all $\mathcal{S} \ni \mathsf{Dict}_j(\cdot)$, evaluate $\mathsf{Dict}_j$ on $\bx^{(1)}, \ldots, \bx^{(N')}$ \;
 \tcp{Hypothesis testing}
\For{$\mathcal{T} \subseteq \mathcal{S}$ and $|\mathcal{T}| =k$} {
 Let $\mathcal{T}= \{\mathsf{Dict}_{i_1}, \ldots, \mathsf{Dict}_{i_k}\}$ \;
  \For{$y \in \{-1,1\}^k$}{Define $ \mathcal{A}_{\mathcal{T},y} = \{\bx^{(s)} : \textrm{for }1 \le s \le N', \ \ 1 \le j \le k \ \textrm{and} \ \mathsf{Dict}_{i_j}(\bx^{(s)}) = y_j\}.$ \; 
  Define $\widehat{\mathsf{Corr}}_{\mathcal{T},y} = \frac{\sum_{\bx^{(s)} \in \mathcal{A}_{\mathcal{T},y}} f(\bx^{(s)})}{M} $}
  Define $\widehat{\mathsf{Corr}}_{\mathcal{T}} = \frac{\sum_y |\widehat{\mathsf{Corr}}_{\mathcal{T},y}|}{2^k}$ and $h_{\mathcal{T}}: \{-1,1\}^k \rightarrow \{-1,1\}$ as 
$
h_{\mathcal{T}}(y) = \mathsf{sign} (\widehat{\mathsf{Corr}}_{\mathcal{T},y} ). 
$ 
 }
 Define $\mathcal{T}^\ast = \arg \max_{\mathcal{T} \subseteq  \mathcal{S} : |\mathcal{T}| =k} \widehat{\mathsf{Corr}}_{\mathcal{T}}$. \; Output $\widehat{\mathsf{Corr}}_{f, \mathcal{S},k} \leftarrow \widehat{\mathsf{Corr}}_{\mathcal{T}^\ast}$ and $h_{\mathcal{S},k} \leftarrow h_{\mathcal{T}^\ast} : \{-1,1\}^k \rightarrow \{-1,1\}$. 
    \caption{\textsf{Find-best-fit}}\label{alg:find-best-fit}
\end{algorithm}
For convenience of the reader, we restate Lemma~\ref{lem:find-best-fit} and then give the proof. 
\rl*
\begin{proof}
The proof of Item~4 is immediate from Step~2 of the algorithm. 
Similarly,  note that the total number of queries made is $|\mathcal{S}| \cdot \mathsf{Poi}(N)$ where $N= O(2^k/\epsilon^2 \cdot (\log(1/\delta) + k^2 + |\mathcal{S}|))$. Item~3 now follows from tail bounds on the Poisson distribution. 

Thus, it just remains to prove Item~5. To prove this, 
consider a fixed $\mathcal{T}$ (from Step~5 of the algorithm \textsf{Find-best-fit}).  Define  $\mathcal{B}_{\mathcal{T},y} = \{x \in \{-1,1\}^n:$ for $1 \le j \le k, \ \mathsf{Dict}_{i_j}(x) = y_j\}$. Observe that $\mathcal{A}_{\mathcal{T}, y}$ can be seen as a sampling of $\mathsf{Poi}(M)$ elements from $\mathcal{B}_{\mathcal{T},y}$. Let us define $\mathsf{Corr}_{\mathcal{T},y} = \mathbf{E}_{\bx \in \mathcal{B}_{\mathcal{T},y}} [f(\bx)]$. Then, observe 
that 
$
{\mathbf{E}}_{\bx^{(1)}, \ldots, \bx^{(N)}} [ \widehat{\mathsf{Corr}}_{\mathcal{T},y} ]  = \mathsf{Corr}_{\mathcal{T},y}. 
$ In fact, by Chernoff-Hoeffding bounds, it follows that 
\[
\Pr_{\bx^{(1)}, \ldots, \bx^{(N)}} [|\mathsf{Corr}_{\mathcal{T},y} - \widehat{\mathsf{Corr}}_{\mathcal{T},y}|> \epsilon/2] \le \frac{2^{-k} \cdot \delta}{10 \cdot 2^{|\mathcal{S}|}}. 
\]
By a union bound, this implies that 
\begin{eqnarray}
\textrm{For all} \ y \in \{-1,1\}^k, \ \ \Pr_{\bx^{(1)}, \ldots, \bx^{(N)}} [|\mathsf{Corr}_{\mathcal{T},y} - \widehat{\mathsf{Corr}}_{\mathcal{T},y}|> \epsilon/2] \le \frac{ \delta}{10 \cdot 2^{|\mathcal{S}|}}. \label{eq:bound-estimate-1}
\end{eqnarray}
This implies that for any subset $\mathcal{T}$, 
\begin{equation}\label{eq:bound-estimate-2}
\Pr_{\bx^{(1)}, \ldots, \bx^{(N)}}[|\mathsf{Corr}_{\mathcal{T},y} - \widehat{\mathsf{Corr}}_{\mathcal{T},y}|> \epsilon/2] \le \frac{\delta}{10 \cdot 2^{|\mathcal{S}|}}. 
\end{equation}

Finally, note that for any subset $\mathcal{T}$,  
\begin{eqnarray*}
\mathbf{E}_{\bx \in \{-1,1\}^n}[h_{\mathcal{T}}(\bx_{i_1}, \ldots, \bx_{i_k}) \cdot f(\bx)]  &=& \mathbf{E}_{\by \in \{-1,1\}^k} \mathbf{E}_{\bx \in \mathcal{A}_{\mathcal{T},y}} [f(\bx) \cdot \mathsf{sign}(\widehat{\mathsf{Corr}}_{\mathcal{T},y})] \\ 
&=& \mathbf{E}_{\by \in \{-1,1\}^k} [ \mathsf{sign}(\widehat{\mathsf{Corr}}_{\mathcal{T},y}) \cdot \mathsf{Corr}_{\mathcal{T},y}]  \\ 
&\ge& \mathbf{E}_{\by \in \{-1,1\}^k} [\mathsf{Corr}_{\mathcal{T},y}] - 2 \max_{\by \in \{-1,1\}^k} [|\mathsf{Corr}_{\mathcal{T},y} - \widehat{\mathsf{Corr}}_{\mathcal{T},y}|]
\end{eqnarray*}
Using (\ref{eq:bound-estimate-1}), we have that 
\begin{equation}\label{eq:bound-estimate-3}
\Pr_{\bx^{(1)}, \ldots, \bx^{(N)}} [ \mathbf{E}_{\bx \in \{-1,1\}^n}[h_{\mathcal{T}}(\bx_{i_1}, \ldots, \bx_{i_k}) \cdot f(\bx)] \ge \mathsf{Corr}_{\mathcal{T}} - \epsilon] \ge 1-\frac{\delta}{10 \cdot 2^{|\mathcal{S}|}}. 
\end{equation}
A union bound over all subsets $\mathcal{T} \subseteq \mathcal{S}$ on (\ref{eq:bound-estimate-2}) and (\ref{eq:bound-estimate-3}) yields Item~5. 
\end{proof}

\end{document}